\documentclass[journal]{IEEEtran}
\usepackage{multirow}
\usepackage{ifpdf}
\usepackage{cite}
\usepackage[pdftex]{graphicx}
\usepackage{amsmath}
\usepackage{amsthm}
\usepackage{array}
\usepackage[caption=false,font=footnotesize]{subfig}
\usepackage{textcomp}
\usepackage[usenames,dvipsnames,svgnames,table]{xcolor}
\usepackage{fixltx2e}
\usepackage{dblfloatfix}
\usepackage[nomarkers]{endfloat}
\usepackage{url}
\usepackage{paralist,esint,amssymb,booktabs,cases,bm,galois}
\newcommand\V[1]{\mathrm{\bf{#1}}}
\newcommand\Tx[1]{\mathrm{#1}}
\newcommand\Se[1]{\mathcal{#1}}
\newcommand\Db[1]{\mathbb{#1}}

\newcommand\Sma[2]{\sum\limits_{#1}^{#2}}
\newcommand\Prd[2]{\prod\limits_{#1}^{#2}}

\newcommand\Nm[1]{\lvert #1\rvert}
\newcommand\Floor[1]{\lfloor #1\rfloor}

\newcommand\MB[1]{\left[#1\right]}
\newcommand\LB[1]{\{#1\}}
\newcommand\SB[1]{\left(#1\right)}

\newcommand{\RN}[1]{\textup{\uppercase\expandafter{\romannumeral#1}}}
\usepackage{cleveref}
\newtheorem{theo}{Theorem}

\newtheorem{lemma}{Lemma}
\newtheorem{exam}{Example}
\newtheorem{Defi}{Definition}
\newtheorem{rem}{Remark}

\usepackage{algorithm}
\usepackage[noend]{algpseudocode}
\makeatletter
\def\BState{\State\hskip-\ALG@thistlm}
\makeatother

\makeatletter
\def\old@comma{,}
\catcode`\,=13
\def,{%
  \ifmmode%
    \old@comma\discretionary{}{}{}%
  \else%
    \old@comma%
  \fi%
}
\makeatother

\begin{document}
\title{Theoretical Bounds and Constructions of Codes in the Generalized Cayley Metric}

\author{Siyi~Yang,~\IEEEmembership{Student Member,~IEEE,}
        Clayton~Schoeny,~\IEEEmembership{Student Member,~IEEE,}
        and~Lara~Dolecek,~\IEEEmembership{Senior Member,~IEEE}
\thanks{This paper was presented in part at the IEEE Information Theory Workshop, Kaohsiung, Taiwan, November, 2017\cite{8277943}.}
\thanks{Siyi Yang is with the Department of Electrical and Computer Engineering, University of California, Los Angeles, CA 90095 USA (email: siyiyang@ucla.edu).} 
\thanks{Clayton Schoeny is with the Department of Electrical and Computer Engineering, University of California, Los Angeles, CA 90095 USA (email: cschoeny@ucla.edu).}      
\thanks{Lara Dolecek is with the Department of Electrical and Computer Engineering, University of California, Los Angeles, CA 90095 USA (email: dolecek@ee.ucla.edu).} 
\thanks{Copyright (c) 2017 IEEE. Personal use of this material is permitted. However, permission to use this material for any other purposes must be obtained from the IEEE by sending a request to pubs-permissions@ieee.org.}
}
\markboth{}%
{Shell \MakeLowercase{\textit{et al.}}: Bare Demo of IEEEtran.cls for IEEE Journals}

\maketitle

\begin{abstract}
Permutation codes have recently garnered substantial research interest due to their potential in various applications, including cloud storage systems, genome resequencing, and flash memories. In this paper, we study the theoretical bounds and constructions of permutation codes in the generalized Cayley metric. The generalized Cayley metric captures the number of generalized transposition errors in a permutation, and subsumes previously studied error types, including transpositions and translocations, without imposing restrictions on the lengths and positions of the translocated segments. Based on the so-called \emph{breakpoint analysis} method proposed by Chee and Vu, we first present a coding framework that leads to order-optimal constructions, thus improving upon the existing constructions that are not order-optimal. We then use this framework to also develop an order-optimal coding scheme that is additionally explicit and systematic.

\end{abstract}


\IEEEpeerreviewmaketitle

\section{Introduction}
\label{sec: introduction}
\IEEEPARstart{G}{eneralized} transposition errors are encountered in various applications, including cloud storage systems, genome resequencing, and flash memories. Cloud storage applications such as Dropbox, OneDrive, iTunes, Google play, etc., are becoming increasingly popular, since they help manage and synchronize data stored across different devices\cite{Music}. When items to be synchronized across are ordered, e.g., in a play list, changes on one device can be viewed as transpositions in the permutation on the other device. In DNA resequencing, released genomes consist of collections of unassembled contigs (a contig is an ordered list of genes in the corresponding genome \cite{SortingCom}), whose organizations evolve over time by undergoing rearrangement operations. Gene order in a chromosome is subject to rearrangements including reversals, transpositions, translocations, block-interchanges, etc. \cite{SortingCom,CHEN2016296}. Generalized transpositions are also encountered in flash memories that utilize rank modulation, a representation in which cells store relative ranks of their charge levels as a permutation. Charge leakage across cells can then be viewed as a sequence of transpositions in the stored permutation. Errors encountered in the applications described above can be appropriately modeled by the generalized Cayley metric for permutation codes, introduced by Chee and Vu, that captures the number of generalized transpositions between two permutations\cite{6875376}.

Permutation codes in the  Kendall-$\tau$ metric and the Ulam metric, along with codes in the Levenshtein metric have been recently actively studied, in \cite{7435295,7089263,7332933}, \cite{6804937,6410424,7282751}, and \cite{7302598,7282754}, respectively. Generalized transposition errors subsume transpositions and translocations that the Kendall-$\tau$ metric and Ulam metric capture, and in particular no restrictions are imposed on the positions and lengths of the translocated segments as in these two metrics. Codes in the generalized Cayley metric were first studied in \cite{6875376} using the \emph{breakpoint analysis}, wherein a coding scheme is constructed based on permutation codes, previously introduced in \cite{6410424}, in the Ulam metric. Let $N$ be the length of the codewords, and $t$ be the maximum number of errors in the generalized Cayley metric. While the coding scheme proposed in \cite{6875376} is explicitly constructive and implementable, the interleaving technique used inevitably incurs a noticeable redundancy of $\Theta\SB{N}$, without even considering the number of errors that the code is able to correct. As we show later, the best possible redundancy of a length-$N$ code that corrects $t$ generalized transposition errors is $\Theta\SB{t\log N}$. When $t$ is $o(\frac{N}{\log N})$, the gap between the redundancy of the existing codes based on interleaving and the optimal redundancy increases with $N$, thus motivating the need to introduce other techniques that are not based on interleaving. We say a length-$N$ code that corrects $t$ generalized transposition errors is order-optimal if the redundancy is $\Theta\SB{t\log N}$.

In order to obtain codes in the generalized Cayley metric that are order-optimal, we present a coding method that is not based on interleaving. The main idea of our coding scheme is to map each permutation of $\LB{1,2,\cdots,N}$ to a unique characteristic set in the Galois field $\Db{F}_q$, where $q$ is a prime number such that $N^2-N<q<2N^2-2N$ and $N$ is the codelength. We prove that the knowledge of the boundaries of the unaltered segments is sufficient for recovering the permutation from its modified version, obtained through generalized transpositions. We exploit the fact that the symmetric difference of the characteristic sets of two distinct permutations corresponds to these boundaries. Given that the number of such boundaries is linearly upper bounded by the number of generalized transpositions, it is sufficient to find permutations with corresponding characteristic sets on $\Db{F}_q$ that have large enough set differences to ensure the desired error correction property. Our proposed method provides a sufficient condition for ensuring the lower bound on the cardinalities of these set differences, which in turn ensures a large enough minimum distance of the resulting code, while the code is order-optimal. Using this approach, we further develop a systematic scheme that is also order-optimal. 

The rest of this paper is organized as follows. In \Cref{sec: measure of distance}, we introduce the basic notation and properties for the generalized Cayley metric and the so-called block permutation metric, which is introduced  for metric embedding. In \Cref{sec: theoboundsrate}, we define the notion of error-correcting codes in these two metrics and derive useful upper and lower bounds on their optimal rates. We prove the optimal rate to be $1-\Theta\SB{\frac{t}{N}}$, and use these results to guide the construction of order-optimal codes. In \Cref{sec: non-syscodes}, we present a method for constructing permutation codes in the generalized Cayley metric. We assign to each permutation of length $N$ a syndrome with elements chosen from a Galois field $\Db{F}_q$, where $q$ is a prime number such that $N^2-N<q<2(N^2-N)$. We prove that the permutations with the same syndrome constitute a codebook, and we prove that the largest one is order-optimal. Based on this method, we then develop a construction for order-optimal systematic permutation codes in the generalized Cayley metric in \Cref{sec: syscodes}. In \Cref{sec: comparison}, we prove that the rates of our proposed codes are higher than those of existing codes based on interleaving, namely, our coding scheme is more rate efficient when $N$ is sufficiently large and $t=o\SB{\frac{N}{\log N}}$. Lastly, we conclude and summarize our main contributions in \Cref{sec: conclusion}. 

\section{Measure of Distance}
\label{sec: measure of distance}

\subsection{Notation}
\label{subsec: notation}
In this paper, we denote by $\left[N\right]$ the set $\{1,2,\cdots,N\}$. We let $\Db{S}_N$ represent the set of all permutations on $\left[N\right]$, where each permutation $\sigma: \left[N\right]\to\left[N\right]$ is a bijection between $\left[N\right]$ and itself. The symbol $\comp$ denotes the composition of functions. Specifically, $\sigma\comp\pi$ denotes the composition of two permutations $\sigma$, $\pi\in\Db{S}_N$, i.e., $\left(\sigma\comp\pi\right)(i)=\sigma\left(\pi\left(i\right)\right)$, $\forall\ i\in\left[N\right]$. We assign a vector $\left(\sigma(1),\sigma(2),\cdots,\sigma(N)\right)$ to each permutation\footnote{We note that this is different from the cycle notation typically used in algebra.} $\sigma\in\Db{S}_N$. Under this notation, we call $e=(1,2,\cdots,N)$ the \emph{identity permutation}. Additionally, $\sigma^{-1}$ is the inverse permutation of $\sigma$. The subsequence of $\sigma$ from position $i$ to $j$, $i\leq j$, is written as $\sigma\left[i;j\right]\triangleq\left(\sigma(i),\sigma(i+1),\cdots,\sigma(j)\right)$. The symbol $\Delta$ refers to the symmetric difference of two sets. Let $\Tx{GCD}\SB{\cdot}$ and $\Tx{LCM}\SB{\cdot}$ be the greatest common divisor and the least common multiple, respectively. The symbol $\equiv$ denotes `congruent modulo'.

\subsection{Generalized Cayley Distance}
\label{subsec: GCD}
A \textbf{\textit{generalized transposition}} $\phi\left(i_1,j_1,i_2,j_2\right)\in\Db{S}_N$, where $i_1\leq j_1<i_2\leq j_2\in \left[ N\right]$, refers to a permutation that is obtained from swapping two segments, $e\left[i_1,j_1\right]$ and $e\left[i_2,j_2\right]$, of the identity permutation \cite{6875376},
\begin{equation}
  \begin{split}
    &\phi\left(i_1,j_1,i_2,j_2\right)\triangleq \left(1,\cdots,i_1-1,i_2,\cdots,j_2,\right.\\
    &\left. j_1+1,\cdots,i_2-1,i_1,\cdots,j_1,j_2+1,\cdots,N\right).
  \end{split}
  \label{eqn: swap}
\end{equation}

Denote the set of all permutations that represent one generalized transposition on any permutation of length $N$ by $\Db{T}_N$. For a given $\pi\in\Db{S}_N$ and $\phi\left(i_1,j_1,i_2,j_2\right)\in\Db{T}_N$, the permutation obtained from swapping the segments $\pi\left[i_1;j_1\right]$ and $\pi\left[i_2;j_2\right]$ is exactly $\pi\comp\phi$, i.e., the permutation,
\begin{equation}
  \begin{split}
    &\left(\pi(1),\cdots,\pi(i_1-1),\pi(i_2),\cdots,\pi(j_2),\pi(j_1+1),\right.\\
    &\left.\cdots,\pi(i_2-1),\pi(i_1),\cdots,\pi(j_1),\pi(j_2+1),\cdots,\pi(N)\right).
  \end{split}
\end{equation}

\begin{exam} \label{exam1}
Let $\pi=\left(3,5,6,7,9,8,1,2,10,4\right)\in\Db{S}_{10}$. Let the underlines mark the subsequences that are swapped by $\phi(2,5,7,8)=\left(1,\underline{7,8},6,\underline{2,3,4,5},9,10\right)$. Then, for $\pi=\left(3,\underline{5,6,7,9},8,\underline{1,2},10,4\right)$, we have:
\begin{align*}
\pi\comp\left(\phi(2,5,7,8)\right)&=\left(3,\underline{1,2},8,\underline{5,6,7,9},10,4\right).
\end{align*}
\end{exam}

\begin{Defi}\label{defi1} \emph{(Generalized Cayley Distance, cf.\cite{6875376})}
The \textbf{generalized Cayley distance} $d_{G}(\pi_1,\pi_2)$ is defined as the minimum number of generalized transpositions that are needed to obtain the permutation $\pi_2$ from $\pi_1$, i.e.,
\begin{equation}
  \begin{split}
    d_{G}(\pi_1,\pi_2)&\triangleq\min\limits_{d}\{\exists\ \phi_1,\phi_2,\cdots,\phi_d\in \Db{T}_N,\text{ s.t., }\\
    &\pi_2=\pi_1\comp\phi_1\comp\phi_2\cdots\comp\phi_d\}.
  \end{split}
  \label{eqn: swapdis}
\end{equation}
\end{Defi}

\begin{rem}\label{rem1}\emph{(cf. \cite{6875376})}. For all $\pi_1,\pi_2,\pi_3\in\Db{S}_N$, the generalized Cayley distance $d_G$ satisfies the following properties: 

\begin{compactenum}
\item (Symmetry) $d_G(\pi_2,\pi_1)=d_G(\pi_1,\pi_2)$.
\item (Left-invariance) $d_G(\pi_3\comp\pi_1,\pi_3\comp\pi_2)=d_G(\pi_1,\pi_2)$.
\item (Triangle Inequality) $d_G(\pi_1,\pi_3)\leq d_G(\pi_1,\pi_2)+d_G(\pi_2,\pi_3)$.
\end{compactenum}

\end{rem}

Notice that the generalized Cayley distance $d_G$ between two permutations is hard to compute, which makes it difficult to construct codes in the generalized Cayley metric. The common method to address the difficulty of specifying the distances between permutations is metric embedding, where one finds another metric that is computable and is of the same order of magnitude as the original metric. We therefore seek to construct codes under the new metric, the so-called \emph{block permutation distance} to be introduced next, and use this construction to specify codes under $d_G$.

\subsection{Block Permutation Distance}
\label{subsec: blockpermdist}
We say a permutation $\pi\in\Db{S}_N$ is \textbf{\emph{minimal}}\footnote{We note that this is different from the usual notion of minimal permutation specified in group theory.} if and only if no consecutive elements in $\pi$ are also consecutive elements in the identity permutation $e$, i.e.,
\begin{equation}
\forall\ 1\leq i<N,\ \pi(i+1)\neq \pi(i)+1.
\end{equation}

The set of all minimal permutations of length $N$ is denoted by $\Db{D}_N$. Next, we define the \emph{block permutation distance} as follows.

\begin{Defi}\label{defi2}
  The \textbf{block permutation distance} $d_B\left(\pi_1,\pi_2\right)$ between two permutations $\pi_1,\pi_2\in\Db{S}_N$ is equal to $d$ if
  \begin{equation}
  \begin{split}
  \pi_1&=\left(\psi_1,\psi_2,\cdots,\psi_{d+1}\right),\\
  \pi_2&=\left(\psi_{\sigma(1)},\psi_{\sigma(2)},\cdots,\psi_{\sigma(d+1)}\right),
  \end{split}
  \label{eqn:blkdis}
  \end{equation}
where $\sigma\in \Db{D}_{d+1}$, $\psi_k=\pi_1\left[i_{k-1}+1:i_k\right]$ for some $0=i_0<i_1\cdots<i_d<i_{d+1}=N$, and $1\leq k\leq d+1$.
\end{Defi}

Note that the block permutation distance between permutations $\pi_1$ and $\pi_2$ is $d$ if and only if $(d+1)$ is the minimum number of blocks the permutation $\pi_1$ needs to be divided into in order to obtain $\pi_2$ through a block-level permutation. Here by block-level permutation we refer to partitioning the original permutation $\pi_1$ into multiple blocks and permuting these blocks.

\begin{exam}\label{examadd} Let $\pi_1=\left(3,5,6,7,9,8,1,2,10,4\right)$, $\pi_2=\left(3,1,2,8,5,6,7,9,10,4\right)$. Define $\psi_i$, $1\leq i\leq 4$, and $\sigma$ as follows,
\begin{align*}
&\psi_1=(3),\psi_2=(5,6,7,9),\psi_3=(8),\psi_4=(1,2),\\
&\psi_5=(10,4),\sigma=(1,4,3,2,5).
\end{align*}
Then,
 \begin{equation}
  \begin{split}
  \pi_1&=\left(\psi_1,\psi_2,\psi_3,\psi_4,\psi_5\right),\\
  \pi_2&=\left(\psi_{\sigma(1)},\psi_{\sigma(2)},\psi_{\sigma(3)},\psi_{\sigma(4)},\psi_{\sigma(5)}\right),
  \end{split}
  \end{equation}
and thus, $d_B(\pi_1,\pi_2)=4$, since $\sigma$ is minimal. This example is in accordance with \Cref{defi2}.

\end{exam}

\begin{lemma}\label{rem2} The block permutation distance $d_B$ also satisfies the properties of symmetry and left-invariance, which are defined in \emph{Remark 1}.
\end{lemma}

\begin{proof} We suppose $\pi_1,\pi_2\in\Db{S}_N$ such that $d_B(\pi_1,\pi_2)=d$. Then, there exist $\sigma\in\Db{S}_{d+1}$, and $\psi_1,\psi_2,\cdots,\psi_{d+1}$ such that $\pi_1=\left(\psi_1,\psi_2,\cdots,\psi_{d+1}\right)$ and $\pi_2=\left(\psi_{\sigma(1)},\psi_{\sigma(2)},\cdots,\psi_{\sigma(d+1)}\right)$.

To prove the symmetry property, we define $\psi'_i=\psi_{\sigma(i)}$ for $1\leq i\leq d+1$, and $\sigma{'}=\sigma^{-1}$. Then, $\sigma'\in\Db{D}_{d+1}$, and
 \begin{equation*}
  \begin{split}
  \pi_2&=\left(\psi'_1,\psi'_2,\cdots,\psi'_{d+1}\right),\\
  \pi_1&=\left(\psi'_{\sigma'(1)},\psi'_{\sigma'(2)},\cdots,\psi'_{\sigma'(d+1)}\right),
  \end{split}
  \end{equation*}
thus, $d_B(\pi_2,\pi_1)=d=d_B(\pi_1,\pi_2)$.

To prove the left-invariance property, suppose the length of $\psi_i$ is $l_i$ and let $\psi_i=\left(\psi_i(1),\psi_i(2),\cdots,\psi_i(l_i)\right)$ for all $1\leq i\leq d+1$. For a given $\pi_3\in\Db{S}_N$, we define $\tilde{\psi}_i=\left(\pi_3\left(\psi_i(1)\right),\pi_3\left(\psi_i(2)\right),\cdots,\pi_3\left(\psi_i(l_i)\right)\right)$, for $1\leq i\leq d+1$. Then,
 \begin{equation*}
  \begin{split}
  \pi_3\comp\pi_1&=\left(\tilde{\psi}_1,\tilde{\psi}_2,\cdots,\tilde{\psi}_{d+1}\right),\\
  \pi_3\comp\pi_2&=\left(\tilde{\psi}_{\sigma(1)},\tilde{\psi}_{\sigma(2)},\cdots,\tilde{\psi}_{\sigma(d+1)}\right).
  \end{split}
  \end{equation*}
Therefore, $d_B(\pi_3\comp\pi_1,\pi_3\comp\pi_2)=d=d_B(\pi_1,\pi_2)$.
\end{proof}

Note that \Cref{defi2} is an implicit representation of $d_B$. Next, we seek to characterize $d_B$ explicitly.

\begin{Defi}\label{defi3}
The \textbf{characteristic set} $A(\pi)$ for any $\pi\in\Db{S}_N$ is defined as the set of all consecutive pairs in $\pi$, i.e.,
\begin{equation}
A(\pi)\triangleq\{\left(\pi(i),\pi(i+1)\right)|1\leq i<N\}.
\label{eqn:charset}
\end{equation}

\end{Defi}
Recall that $e$ refers to the identity permutation.

\begin{Defi}\label{defi4}
The \textbf{block permutation weight} $w_B\left(\pi\right)$ is defined as the number of consecutive pairs in $\pi$ that do not belong to $A(e)$ ($w_B$ is exactly the number of so-called  breakpoints in \cite{6875376}), i.e.,
\begin{equation}
    w_B\left(\pi\right)\triangleq\lvert A(\pi)\setminus A(e)\rvert.
    \label{eqn:blkweight}
\end{equation}
 
\end{Defi}

\Cref{lemma1} and \Cref{rem3} state explicit representations of the block permutation distance $d_B$ by the characteristic set and the block permutation weight, respectively, and will be used later in the paper to establish our main result.

\begin{lemma}\label{lemma1}
For all $\pi_1,\pi_2\in\Db{S}_N$, 
\begin{equation}
d_B(\pi_1,\pi_2)=\lvert A(\pi_2)\setminus A(\pi_1)\rvert=\lvert A(\pi_1)\setminus A(\pi_2)\rvert.
\label{eqn:charsetdis}
\end{equation}
\end{lemma}

\begin{proof} The proof is in Appendix \ref{prooflemma1}.
\end{proof}

\begin{rem}\label{rem3} From \Cref{lemma1} and \Cref{defi4}, it is obvious that
\begin{equation}
    w_B\left(\pi\right)=d_B(e,\pi)=d_B(\pi,e).
\end{equation}

For all $\pi_1,\pi_2\in\Db{S}_N$, it follows immediately from the left-invariance property of $d_B$ and (\ref{eqn:blkweight}) that 
  \begin{equation}
    d_B\left(\pi_1,\pi_2\right)=w_B\left(\pi_1^{-1}\comp\pi_2\right).
    \label{eqn: blkdis}
  \end{equation} 

\end{rem}

In \Cref{exam2}, we show how to compute the block permutation distance of two permutations from their characteristic sets, as it is indicated in \Cref{lemma1}.

\begin{exam}\label{exam2} For $\pi_1,\pi_2$ specified in \Cref{examadd},
\begin{equation*}
\begin{split}
A(\pi_1)&=\{(3, 5), (5, 6), (6, 7), (7, 9),(9, 8),\\
& (8, 1), (1, 2), (2, 10), (10, 4)\},\\
A(\pi_2)&=\{(3, 1), (1, 2), (2, 8), (8, 5),(5, 6), \\
&(6, 7), (7, 9), (9, 10), (10, 4)\}.\\
\end{split}
\end{equation*}
Therefore,
\begin{equation*}
\begin{split}
\lvert A(\pi_1)\setminus A(\pi_2)\rvert=&\lvert \{(3, 5), (9, 8), (8, 1), (2, 10)\}\rvert\\
=&4=d_B(\pi_1,\pi_2).
\end{split}
\end{equation*}
This example is in accordance with \Cref{lemma1}.

\end{exam}

\subsection{Metric Embedding}
\label{subsec: subsectionmetricembedd}
The generalized Cayley distance is difficult to compute, whereas the block permutation distance can be computed efficiently. Therefore, it is easier to check whether two distinct candidate codewords in a codebook meet the minimum requirement on the block permutation distance, than it is to check whether they meet the minimum requirement on the generalized Cayley distance. In light of this observation, in the next section, we apply metric embedding to transform the problem of code design in $d_{G}$ into that in $d_B$, which is easier to deal with, using the following results.
 
\begin{lemma}\label{lemma2} For all $\pi_1,\pi_2\in\Db{S}_N$, the following inequality holds,
  \begin{equation}
    w_B\left(\pi_1\comp\pi_2\right)\leq w_B\left(\pi_1\right)+w_B\left(\pi_2\right).
    \label{eqn:triblkpermw}
  \end{equation}

\end{lemma}

\begin{proof} The proof is in Appendix \ref{prooflemma2}.
\end{proof}

\begin{rem}\label{rem4} It follows immediately from equation (\ref{eqn: blkdis}) and \Cref{lemma2} that the block permutation distance satisfies the triangle inequality, i.e., $\forall$ $\pi_1,\pi_2,\pi_3\in\Db{S}_N$, 
\begin{equation}
d_B(\pi_1,\pi_3)\leq d_B(\pi_1,\pi_2)+d_B(\pi_2,\pi_3).
\label{dGT}
\end{equation}
\end{rem}

From \Cref{lemma2} and the definitions of the generalized Cayley metric and the block permutation metric, we observe the following relation between $d_B$ and $d_{G}$. This result is used later in \Cref{sec: non-syscodes}.

\begin{lemma}\label{lemma3} For all $\pi_1,\pi_2\in\Db{S}_N$, the following inequality holds,
\begin{equation}
d_{G}\left(\pi_1,\pi_2\right)\leq d_B\left(\pi_1,\pi_2\right)\leq 4 d_{G}\left(\pi_1,\pi_2\right).
\label{eqn:metricemb}
\end{equation}
\end{lemma}

\begin{proof}

To prove the upper bound, we consider two arbitrary permutations $\pi_1$, $\pi_2\in \Db{S}_N$, and let $k=d_G(\pi_1,\pi_2)$. We know from definitions of a generalized transposition and the block permutation weight that for any generalized transposition $\phi\in\Db{T}_N$ (recall that $\Db{T}_N$ is defined at the beginning of \Cref{subsec: GCD} as the set of all permutations that represent a generalized transposition in permutations of length $N$), the following inequality holds,
\begin{equation}
  w_B\left(\phi\right)\leq 4.
  \label{eqn: swapweight}
\end{equation}

From the definition of the generalized Cayley metric and $d_G(\pi_1,\pi_2)=k$, it follows that for some $\phi_1,\phi_2,\cdots,\phi_k\in \Db{T}_N$,
\[\pi_2=\pi_1\comp\phi_1\comp\phi_2\cdots\comp\phi_k.\]
Then, from \Cref{lemma2} and (\ref{eqn: swapweight}), 
\begin{equation*}
  \begin{split}
    d_B\left(\pi_1,\pi_2\right)=&w_B\left(\pi_1^{-1}\comp\pi_2\right)\\
    =&w_B\left(\phi_1\comp\phi_2\comp\cdots\comp\phi_{k}\right)\\
    \leq&\sum\limits_{i=1}^{k}w_B\left(\phi_i\right)\\
    \leq&4k=4d_G\left(\pi_1,\pi_2\right).
  \end{split}
\end{equation*}
The upper bound is proved.

The lower bound is trivially  attained when $\pi_1=\pi_2$. When $\pi_1$ and $\pi_2$ are distinct, it follows that $d_B(\pi_1,\pi_2)=d$ for some positive integer $d$. Then, according to the definition of the block permutation distance, there exists a minimal permutation $\sigma$ (minimal permutation is defined in \Cref{subsec: blockpermdist} as a permutation where no consecutive elements in $\sigma$ are also consecutive elements in the identity permutation) and a partition $\{\psi_i\}_{i=1}^{d+1}$ of $\pi_1$ such that, $\pi_1=\left(\psi_1,\psi_2,\cdots,\psi_{d+1}\right)$, and $\pi_2=\left(\psi_{\sigma(1)},\psi_{\sigma(2)},\cdots,\psi_{\sigma(d+1)}\right)$.

Next, suppose $l_0$ is the smallest index $l$ such that $\sigma(l)\neq l$, $1\leq l \leq d+1$ (the assumption that $\pi_1\neq \pi_2$ ensures the existence of $l_0$). Let $k_0=\sigma^{-1}(l_0)$, then $k_0>l_0$. Let $\phi_1$ represent the generalized transposition that swaps the subsequences $\left(\psi_{\sigma(l_0)},\psi_{\sigma(l_0+1)},\cdots,\psi_{\sigma(k_0-1)}\right)$ and $\psi_{\sigma(k_0)}=\psi_{l_0}$ in $\pi_2$. Let $\pi_2^{(1)}=\pi_2\comp\phi_1$ and $\sigma^{(1)}=(1,2,\cdots,l_0,\sigma(l_0),\sigma(l_0+1),\cdots,\sigma(k_0-1),\sigma(k_0+1),\cdots,\sigma(d+1))$. Then,
  \[\pi_2^{(1)}=\left(\psi_{\sigma^{(1)}(1)},\psi_{\sigma^{(1)}(2)},\cdots,\psi_{\sigma^{(1)}(d+1)}\right).\]
If $\pi_2^{(1)}=\pi_1$, then $\pi_1=\pi_2\comp\phi_1$. Otherwise let $l_1$ be the smallest index $l$ such that $\sigma^{(1)}(l)\neq l$, $1\leq l\leq d+1$, then $l_1>l_0$ holds true.

Following this procedure, one can find a series of generalized transpositions $\phi_1,\phi_2,\cdots,\phi_m$, $1\leq m\leq d$, sequentially, such that $\pi_2\comp\phi_1\comp\phi_2\comp\cdots\comp\phi_m=\pi_1$. Suppose $\phi_1,\phi_2,\cdots,\phi_i$ are found for some $i$, $1\leq i\leq d$. Let $\pi_2^{(i)}=\pi_2\comp\phi_1\comp\phi_2\comp\cdots\comp\phi_i=\SB{\psi_{\sigma^{(i)}(1)},\psi_{\sigma^{(i)}(2)},\cdots,\psi_{\sigma^{(i)}(d+1)}}$. If $\pi_2^{(i)}=\pi_1$, then $\pi_1=\pi_2\comp\phi_1\comp\phi_2\comp\cdots\comp\phi_i$, and we have established the desired composition. Otherwise, we let $l_i$ be the smallest index such that $\sigma^{(i)}(l_i)\neq l_i$. Suppose $k_i=\SB{\sigma^{(i)}}^{-1}(l_i)$, and it follows that $k_i>l_i$. Denote the generalized transposition that swaps the subsequences $\SB{\psi_{\sigma^{(i)}(l_i)},\psi_{\sigma^{(i)}(2)},\cdots,\psi_{\sigma^{(i)}(k_i-1)}}$ and $\psi_{\sigma^{(i)}(k_i)}=\psi_{l_i}$ in $\pi_2^{(i)}$ by $\phi_{i+1}$. Let $\pi_2^{(i+1)}=\pi_2^{(i)}\comp\phi_{i+1}$, and $\sigma^{(i+1)}=(1,2,\cdots,l_i,\sigma^{(i)}(l_i),\sigma^{(i)}(l_i+1),\cdots,\sigma^{(i)}(k_i-1),\sigma^{(i)}(k_i+1),\cdots,\sigma^{(i)}(d+1))$. Then,
 \[\pi_2^{(i+1)}=\left(\psi_{\sigma^{(i+1)}(1)},\psi_{\sigma^{(i+1)}(2)},\cdots,\psi_{\sigma^{(i+1)}(d+1)}\right).\]

Finally, one finds the smallest integer $m$ such that $\pi_2^{(m)}=\pi_1$. In this procedure, $l_0,\cdots,l_{m-1}$ are obtained sequentially, where $1<l_0<l_1<\cdots<l_{m-1}$. We also know that $l_{m-1}\leq d$, otherwise if $l_{m-1}=d+1$, then $\sigma^{(m-1)}(i)=i$ holds true for all $1\leq i\leq d$, and $\sigma^{(m-1)}(d+1)\neq d+1$, which leads to a contradiction. Therefore, $d\geq l_{m-1}>\cdots>l_0\geq 1$, which implies that $m\leq d$. Note that $\pi_1=\pi_2\comp\phi_1\comp\cdots\comp\phi_m$, from which $d_G(\pi_1,\pi_2)\leq m\leq d=d_B(\pi_1,\pi_2)$ follows. The lower bound is proved.
\end{proof}

\section{Theoretical Bounds on the Code Rate}
\label{sec: theoboundsrate}
A subset $\Se{C}_{G}\left(N,t\right)$ of $\Db{S}_N$ is called a $\bm{t}$\textbf{\textit{-generalized Cayley code}} if it can correct $t$ generalized transposition errors. Any $t$-generalized Cayley code has the minimum generalized Cayley distance $d_{G,min}\geq 2t+1$. Similarly, a subset $\Se{C}_B\left(N,t\right)$ of $\Db{S}_N$ is called a $\bm{t}$\textbf{\textit{-block permutation code}} if its minimum block permutation distance $d_{B,min}\geq 2t+1$. For any permutation code $\Se{C}\subset \Db{S}_N$, denote the rate of $\Se{C}$ by $R(\Se{C})$. Then, the following equation holds true,
\begin{equation}
\begin{split}
R(\Se{C})=\frac{\log \lvert\Se{C}\left(N,t\right)\rvert}{\log N!}.\\
\end{split}
\label{R}
\end{equation}
In the remainder of this paper, the logarithm base is always $2$ unless it is explicitly specified with a different base.

Let $\Se{C}_{G,opt}\left(N,t\right)$ and $\Se{C}_{B,opt}\left(N,t\right)$ be $t$-generalized Cayley codes and $t$-block permutation codes with the optimal rates, denoted by $R_{G,opt}(N,t)$ and $R_{B,opt}(N,t)$, respectively. We next derive the lower bounds and the upper bounds of $R_{G,opt}\left(N,t\right)$ and $R_{B,opt}\left(N,t\right)$.

For each $\pi\in\Db{S}_N$, we define the \textbf{\emph{generalized Cayley ball}} $B_{G}(N,t,\pi)$ of radius $t$ centered at $\pi$ to be the set of all permutations in $\Db{S}_N$ that have a generalized Cayley distance from $\pi$ not exceeding $t$. We know from the left-invariance property of $d_{G}$ that the cardinality of $B_{G}(N,t,\pi)$ is independent of $\pi$; we denote $\lvert B_G(N,t,\pi)\rvert$ as $b_{G}(N,t)$. The \textbf{\emph{block permutation ball}} $B_B(N,t,\pi)$ and the corresponding ball-size $b_B(N,t)$ are similarly defined. 

We derive the lower and upper bounds of $b_B(N,t)$ and $b_{G}(N,t)$ in the following two lemmas, respectively. We build on these results and \Cref{lemma6} to compute the bounds of the rates of optimal codes in $d_G$ and $d_B$, proving that the optimal redundancy is $\Theta(\frac{t}{N})$ in both of the two metrics.

\begin{lemma}\label{lemma4}
For all $N\in\Db{N}^{*}$, $t\leq N-\sqrt{N}-1$, $b_B(N,t)$ is bounded by the following inequality:
\begin{equation}
\prod\limits_{k=1}^{t}(N-k)\leq b_B(N,t)\leq \prod\limits_{k=0}^{t}(N-k).\\
\label{eqn:lemmablk}
\end{equation}

\end{lemma}

\begin{proof} The proof is in Appendix \ref{prooflemma4}.
\end{proof}

\begin{lemma}\label{lemma5}
For all $N\in\Db{N}^{*}$, $t\leq \min\{N-\sqrt{N}-1,\frac{N-1}{4}\}$, $b_{G}(N,t)$ is bounded as follows:
\begin{equation}
\prod\limits_{k=1}^{t}(N-k)\leq b_{G}(N,t)\leq \prod\limits_{k=0}^{4t}(N-k).\\
\label{eqn:lemmaswap}
\end{equation}
\end{lemma}

\begin{proof} The proof is in Appendix \ref{prooflemma5}.
\end{proof}

As the metrics $d_B$ and $d_G$ both satisfy the triangle inequality, the cardinalities of the optimal codes $\Se{C}_{B,opt}(N,t)$ and $\Se{C}_{G,opt}(N,t)$ are bounded as follows,
\begin{equation}
\begin{split}
\frac{N!}{b_B(N,2t)}\leq \lvert &\Se{C}_{B,opt}\left(N,t\right)\rvert\leq \frac{N!}{b_B(N,t)},\\
\frac{N!}{b_G(N,2t)}\leq \lvert &\Se{C}_{G,opt}\left(N,t\right)\rvert\leq \frac{N!}{b_G(N,t)}.
\label{eqn:blkpermbound}
\end{split}
\end{equation}

According to \cite[(1)-(2)]{10.2307/2308012}, for all $N\in\Db{N}^{*}$,
\begin{equation}\label{stirling}
N!=\sqrt{2\pi} N^{N+1/2} e^{-N}\cdot e^{r_N},
\end{equation}
where
\begin{equation}\label{rn}
\frac{1}{12N+1}<r_N<\frac{1}{12N}.
\end{equation}
From (\ref{stirling}) and (\ref{rn}), \Cref{lemma6} follows.

\begin{lemma}\label{lemma6} For all $N\in\Db{N}^{*}$, it follows that
\begin{equation*}
\begin{split}
(N+\frac{1}{2})\log N -&(\log e)N<\sum\limits_{n=1}^{N} \log n\\
&<(N+\frac{1}{2})\log N-(\log e)N+2.
\end{split}
\end{equation*}
\end{lemma}

We now state the main result of this section.

\begin{theo}\label{theo1} For any $t,N\in\Db{N}^{*}$, $t\leq \min\{N-\sqrt{N}-1,\frac{N-1}{4}\}$ and $N\geq 9$, the optimal rates $R_{B,opt}\left(N,t\right),R_{G,opt}\left(N,t\right)$ satisfy the following inequalities,
\begin{equation}
\begin{split}
1-c\cdot\frac{2t+1}{N}\leq&R_{B,opt}\left(N,t\right)\leq 1-\frac{t}{N},\\
1-c\cdot\frac{8t+1}{N}\leq&R_{G,opt}\left(N,t\right)\leq 1-\frac{t}{N},
\end{split}
\end{equation}
where $c=1+\frac{2\log e}{\log N}$.

\end{theo}

\begin{proof}
From (\ref{R}) and (\ref{eqn:blkpermbound}), it follows that 

\begin{equation}
\begin{split}
1-\frac{\log b_B(N,2t)}{\log N!}\leq R_{B,opt}(N,t)&\leq 1-\frac{\log b_B(N,t)}{\log N!},\\
1-\frac{\log b_G(N,2t)}{\log N!}\leq R_{G,opt}(N,t)&\leq 1-\frac{\log b_G(N,t)}{\log N!}.\\
\end{split}
\label{Rlog}
\end{equation}

By applying \Cref{lemma4} and \Cref{lemma6} to (\ref{Rlog}), when $\min\{N-\sqrt{N}-1,\frac{N-1}{4}\}\geq t\geq 1$ and $N\geq 9$, it follows that
\begin{equation}
\begin{split}
R_{B,opt}(N,t)&\geq 1-\frac{\log \left[\prod\limits_{k=0}^{2t}(N-k)\right]}{\log N!}\\
&>1-\frac{(2t+1)\log N}{(N+\frac{1}{2})\log N -(\log e)N}\\
&>1-\frac{(2t+1)\log N}{N(\log N -\log e)}\\
&>1-\frac{2t+1}{N}\left(1+\frac{2\log e}{\log N}\right),\\
\end{split}
\label{eqn:add3}
\end{equation}

and
\begin{equation}
\begin{split}
 &R_{B,opt}(N,t)\\
 \leq& 1-\frac{\log \left[\prod\limits_{k=1}^{t}(N-k)\right]}{\log N!}\\
 =& 1-\frac{\frac{1}{2}\sum\limits_{k=1}^{t}\SB{\log (N-k)+\log (N-t-1+k)}}{\log N!}\\
 \leq& 1-\frac{\frac{t}{2}\log \SB{(N-1)(N-t)} }{(N+\frac{1}{2})\log N-(\log e)N+2}\\
 \leq& 1-\frac{\frac{t}{2}\log \SB{(N-1)(N-\frac{N-1}{4})} }{(N+\frac{1}{2})\log N-(\log e)N+2}\\
 \leq& 1-\frac{\frac{t}{2}\log \SB{\frac{N^2}{2}} }{(N+\frac{1}{2})\log N-(\log e)N+2}\\
 \leq& 1-\frac{t(\log N -\frac{1}{2})}{N\log N-\frac{1}{2}N}\\
 =&1-\frac{t}{N}.\\
 \end{split}
 \label{eqn:add4}
\end{equation}

Similarly, by applying \Cref{lemma5} and \Cref{lemma6} to (\ref{Rlog}), when $\min\{N-\sqrt{N}-1,\frac{N-1}{4}\}\geq t\geq 1$ and $N\geq 9$, it follows that
\begin{equation}
\begin{split}
R_{G,opt}(N,t)&\geq 1-\frac{\log \left[\prod\limits_{k=0}^{\min\{8t,N-1\}}(N-k)\right]}{\log N!}\\
&>1-\frac{(8t+1)\log N}{(N+\frac{1}{2})\log N -(\log e)N}\\
&>1-\frac{(8t+1)\log N}{N\log N -(\log e)N}\\
&>1-\frac{8t+1}{N}\left(1+\frac{2\log e}{\log N}\right),\\
\end{split}
\label{eqn:add2}
\end{equation}
and
\begin{equation}
\begin{split}
R_{G,opt}(N,t)&\leq 1-\frac{\log \left[\prod\limits_{k=1}^{t}(N-k)\right]}{\log N!}\\
&\leq 1-\frac{t}{N}.
\end{split}
\label{eqn:add1}
\end{equation}

The theorem is proved.
\end{proof}

Inequalities (\ref{eqn:add3})-(\ref{eqn:add1}) indicate that $R=1-\Theta\left(\frac{t}{N}\right)$ is the rate of the $t$-generalized Cayley codes and the $t$-block permutation codes that are order-optimal.

\section{Non-Systematic Permutation Codes in the Generalized Cayley Metric}
\label{sec: non-syscodes}
We studied the optimal rates of $t$-generalized Cayley Codes and $t$-block permutation codes in the previous section. We now seek constructions of order-optimal codes in these metrics. We know from \Cref{lemma3} that any $4t$-block permutation code is also a $t$-generalized Cayley code. In the sequel, we thus focus on the construction of order-optimal $t$-block permutation codes, which is sufficient for obtaining order-optimal generalized Cayley codes.

In \Cref{subsec: nonsysencode}, we present a construction of order-optimal $t$-block permutation codes (\Cref{theo2}). We then develop a decoding scheme for the proposed codes in \Cref{subsec: nonsysdecode}.

\subsection{Encoding Scheme}
\label{subsec: nonsysencode}
Denote the set of all ordered pairs of non-identical elements from $\left[N\right]$ by $P$; then $|P|=N^2-N$. Suppose $q$ is a prime number such that $q\geq |P|$. From \emph{Bertrand's postulate} \cite{ramanujan1919proof}, one can always find a prime number $q$ such that $|P|\leq q\leq 2|P|$.

Let $\upsilon:\ P\to \Db{F}_q$ be an arbitrary injection from $P$ to $\Db{F}_q$, where $\Db{F}_q$ is a Galois field of order $q$. Let $\Se{P}(\Db{F}_q)$ represent the set of all the subsets of $\mathbb{F}_q$ with cardinality $N-1$. We define an injection $\nu: \Db{S}_N\to \Se{P}(\Db{F}_q)$ as follows:

\begin{equation}
\nu(\pi)\triangleq\{\upsilon(p)|p\in A(\pi)\}.
\label{eqn:defig}
\end{equation}
Then, $\nu$ is invertible, namely, one is able to compute $\pi$ based on $\nu(\pi)$.

We then define a class of functions $\alpha^{(q,d)}:\ \Db{S}_N\to \Db{F}_q^{2d-1}$, as follows:
\begin{equation}\label{eqn:defialpha1}
\alpha^{(q,d)}(\pi)\triangleq\left(\alpha_1,\alpha_2,\cdots,\alpha_{2d-1}\right),
\end{equation}
where 
\begin{eqnarray}\label{alphadefi}
\left\{
\begin{array}{rcll}
\alpha_1&\equiv &\sum\limits_{b\in \nu(\pi)}b&\mod{q},\\
\alpha_2&\equiv&\sum\limits_{b\in \nu(\pi)} b^2&\mod{q},\\
&\vdots&&\\
\alpha_{2d-1}&\equiv&\sum\limits_{b\in \nu(\pi)} b^{2d-1}&\mod{q}.\\
\end{array}
\right.
\end{eqnarray}

The following \Cref{algo: Encode1} describes the main steps of the proposed encoding scheme, the correctness of which can be verified by \Cref{lemma7} and \Cref{theo2}.

\begin{algorithm}
\caption{Encoding Scheme}\label{algo: Encode1}
\begin{algorithmic}[1]
\Require 
\Statex Minimum block permutation distance: $2t+1$;
\Statex Codelength: $N$;
\Statex Alphabet size: $q$, where $q$ is a prime number such that $N^2-N\leq q<2(N^2-N)$;
\Ensure
\Statex Codebook $\Se{C}$ of a $t$-block permutation code;
\State For each $\pi\in \Db{S}_N$, compute $A(\pi)$, $\nu(\pi)$, and its syndrome $\alpha^{(q,2t)}(\pi)$ ($\alpha^{(q,2t)}(\pi)\in\Db{F}_q^{4t-1}$), sequentially, where $A(\pi)$, $\nu(\pi)$, $\alpha^{(q,d)}$ are defined in \Cref{defi3}, (\ref{eqn:defig}), (\ref{eqn:defialpha1}) and (\ref{alphadefi}), respectively;
\State For each $\bm{\alpha}\in \Db{F}_q^{4t-1}$, denote the set consisting of all permutations with the syndrome $\bm{\alpha}$ by $\Se{C}_{\bm{\alpha}}(N,t)$;
\State Find $\bm{\alpha}$ such that $\Se{C}_{\bm{\alpha}}(N,t)$ is of the maximum cardinality;
\State \textbf{return} $\Se{C}=\Se{C}_{\bm{\alpha}}(N,t)$.
\end{algorithmic}
\end{algorithm}

The following \Cref{lemma7} states that the cardinality of the symmetric difference of $\nu(\pi_1),\nu(\pi_2)$ for any two distinct permutation $\pi_1,\pi_2\in\Db{S}_N$ is greater than $2d$ if their syndromes $\alpha^{(q,d)}(\pi_1)$ and $\alpha^{(q,d)}(\pi_2)$ are identical. Therefore, their block permutation distance is greater than $d$ based on \Cref{lemma1}. This lemma will be repeatedly used in the rest of the paper for the constructions of order-optimal permutation codes in the block permutation distance.

\begin{lemma}\label{lemma7}
For all $\pi_1$, $\pi_2\in\Db{S}_N$ such that $\pi_1\neq \pi_2$, if $\alpha^{(q,d)}(\pi_1)=\alpha^{(q,d)}(\pi_2)$, then,
\begin{equation}
\lvert \nu(\pi_1)\Delta \nu(\pi_2) \rvert>2d.
\label{eqn: diff}
\end{equation}

\end{lemma}

\begin{proof} The proof is in Appendix \ref{prooflemma7}.
\end{proof}

Note that the function $\alpha^{(q,2t)}$ induces a map from $\Db{S}_N$ to $\Db{F}_q^{4t-1}$ and divides $\Db{S}_N$ into $q^{4t-1}$ subsets based on their syndromes $\bm{\alpha}=(\alpha_1,\alpha_2,\cdots,\alpha_{4t-1})$. We next prove that each such subset is a $t$-block permutation code, which is stated as the following theorem.

\begin{theo}\label{theo2}
For all $\bm{\alpha}\in\Db{F}_q^{4t-1}$, suppose:
\begin{equation}
\Se{C}_{\bm{\alpha}}(N,t)=\{\pi|\pi\in\Db{S}_N,\ \alpha^{(q,2t)}(\pi)=\bm{\alpha}\},
\end{equation}
where $\alpha^{(q,2t)}$ is defined in (\ref{eqn:defialpha1}) and (\ref{alphadefi}). Then $\forall$ $\pi_1,\pi_2\in \Se{C}_{\bm{\alpha}}(N,t)$, $\pi_1\neq\pi_2$, the following inequality holds,
\begin{equation}
d_B(\pi_1,\pi_2)\geq 2t+1.
\end{equation}
\end{theo}

\begin{proof}
Let $d=2t$ in \Cref{lemma7} and \Cref{lemma1}. Then, 
\begin{equation}
\begin{split}
d_B(\pi_1,\pi_2)&=\frac{1}{2}\lvert A(\pi_1)\Delta A(\pi_2)\rvert\\
&=\frac{1}{2}\lvert \nu(\pi_1)\Delta \nu(\pi_2) \rvert\\
&>\frac{1}{2}(2\cdot 2t)=2t,
\end{split}
\label{eqn:disdiff}
\end{equation}
where $\Delta$ refers to the symmetric difference of sets. 
\end{proof}

\begin{exam}\label{exam:add2} Suppose $N=10$, $t=2$, $q=97>10^2-10$. Define $\upsilon(i,j)$ for all $i\neq j\in\left[10\right]$ as follows:
\begin{equation*}
\upsilon(i,j)=10(i-1)+j-1.
\end{equation*} 

Let $\pi_1=(10,9,8,7,6,5,4,3,2,1)$, and $\pi_2=(9,6,5,8,2,4,7,3,10,1)$. Suppose $\bm{\alpha}=(83,28,80,77,40,3,88)$. Then,
\begin{equation*}
\alpha^{(q,2t)}(\pi_1)=\alpha^{(q,2t)}(\pi_2)=\bm{\alpha}.
\end{equation*}
Observe that $d_B(\pi_1,\pi_2)=8>4=2t$. This example is in accordance with \Cref{theo2}.
\end{exam}

\Cref{theo2} implies that $\{\Se{C}_{\bm{\alpha}}(N,t):\bm{\alpha}\in\Db{F}_q^{4t-1}\}$ is a partition of $\Db{S}_N$, where each component $\Se{C}_{\bm{\alpha}}(N,t)$ is a $t$-block permutation code indexed by $\bm{\alpha}$. Suppose $\Se{C}_{\bm{\alpha}_{\Tx{max}}}(N,t)$ is the one with the maximal cardinality, whose syndrome is $\bm{\alpha}_{\Tx{max}}$. It follows from the \textit{Pigeonhole Principle} that:
\begin{equation}
\lvert \Se{C}_{\bm{\alpha}_{\Tx{max}}}(N,t)\rvert\geq \frac{N!}{\lvert \Db{F}_q^{4t-1} \rvert}= \frac{N!}{q^{4t-1}}.
\end{equation} 

Denote the rate of $\Se{C}_{\bm{\alpha}_{\Tx{max}}}(N,t)$ by $R(\Se{C}_1)$. Given that $N^2-N=|P|\leq q <2|P|=2N^2-2N<2N^2$, it follows from \Cref{lemma6} that for $N>e^2$ (note that here $e$ refers to the base of the natural logarithm),
\begin{equation}
\begin{split}
R(\Se{C}_1)&\geq 1-\frac{4t\log q}{\log N!}> 1-\frac{8t\log N+4t}{\log N!}\\
&> 1-\frac{8t(\log N+\frac{1}{2})}{(N+\frac{1}{2})\log N-(\log e)N}\\
&> 1-\frac{8t}{N}\left(\frac{\log N+\frac{1}{2}}{\log N-\log e}\right)\\
&=1-\frac{8t}{N}\MB{1+\frac{\frac{1}{2}+\log e}{\log N}\SB{1+\frac{\log e}{\log N-\log e}}}\\
&>1-\frac{8t}{N}\left(1+\frac{2\log e+1}{\log N}\right).
\end{split}
\end{equation}

Then, $\Se{C}_{\bm{\alpha}_{\Tx{max}}}(N,t)$ is an order-optimal $t$-block permutation code.

\subsection{Decoding Scheme}
\label{subsec: nonsysdecode}
In \Cref{subsec: nonsysencode}, we map each permutation $\pi\in\Db{S}_N$ to a unique set $\nu(\pi)\in\Se{P}(\Db{F}_q)$ as defined in equation (\ref{eqn:defig}), where $N^2-N\leq q\leq 2N^2-2N$ and $\Se{P}(\Db{F}_q)$ represents the set consisting of all subsets of $\Db{F}_q$ with cardinality $N-1$. Suppose the transmitter sends $\pi\in\Db{S}_N$ and the receiver receives $\pi'$, where $d_{G}(\pi,\pi')\leq t$. In the decoding scheme, our objective is to compute $\nu(\pi)$ from the a priori $\bm{\alpha}$ and the received permutation $\pi'$. The strategy is, for each set $B\in\Se{P}(\Db{F}_q)$, map $B$ to a polynomial $f(X;B)$ defined as follows:
\begin{equation}
f(X;B)\triangleq\prod\limits_{b\in B} \left(X+b\right).
\label{esp}
\end{equation}

We call $f(X;B)$ the \textbf{\emph{characteristic function}} of set $B$. All the polynomials as well as the polynomial operations are defined on $\Db{F}_q$. Let $a_i^B$, $0\leq i\leq N-1$, represent the coefficients of $X^{N-1-i}$ in $f(X;B)$. Then, $a_0^B=1$.

Given the a priori agreement on the codebook, i.e., the choice of $\bm{\alpha}$, and the received permutation $\pi'$, the value of the first $4t$ coefficients of $f(X;B)$, $f(X;B')$ can be computed, where $B=\nu(\pi)$ and $B'=\nu(\pi')$, as we shall shortly show. We then use these coefficients to derive $\nu(\pi)$. This coding strategy bears resemblance to that proposed in \cite{Dolecek}, the key difference being that the coefficients of the polynomials we discussed are partially known, thus making our decoding scheme more complicated, whereas those in \cite{Dolecek} are fully known. 

Note that $a_i^B, 1\leq i\leq N-1$, in (\ref{eqn:ni}) is the $i$-th elementary symmetric polynomial of the elements in $B$. Also note that the $i$-th component $\alpha_i, 1\leq i\leq 4t-1$, of the value $\bm{\alpha}=\alpha^{(q,2t)}(\pi)$ is exactly the $i$-th power sum of the elements in $B=\nu(\pi)$. We know from \emph{Newton's identities}\cite{zeilberger1984combinatorial} that there exists a bijection between the $(4t-1)$ power sums and the first $(4t-1)$ elementary symmetric polynomials of elements in $B$, as described below:

\begin{equation}
\begin{cases}
&a_0^B=1,\\
&a_1^B=\alpha_1,\\
&a_2^B=2^{-1}(a_1^B\alpha_1-\alpha_2),\\
&a_3^B=3^{-1}(a_2^B\alpha_1-a_1^B\alpha_2+\alpha_3),\\
&\vdots\\
&a_{4t-1}^B=(4t-1)^{-1}(a_{4t-2}^B\alpha_1-a_{4t-3}^B\alpha_2+\cdots+\alpha_{4t-1}).
\end{cases}
\label{eqn:ni}
\end{equation}

Denote $a_i^B$, $a_i^{B'}$ by $a_i$, $a'_i$, $0\leq i\leq N-1$, respectively, for simplicity. Let $r(B)=(a_1,a_2,\cdots,a_{4t-1})$, $r(B')=(a_1',a_2',\cdots,a_{4t-1}')$. The receiver uses the a priori $\bm{\alpha}$ to compute $r(B)$ and to derive $r(B')$ from $B'$, where $B=\nu(\pi)$ and $B'=\nu(\pi')$. Note that $\pi$ can be computed from $B=\nu(\pi)$ since $\nu$ is an injection from $\Db{S}_N$ to $\Se{P}(\Db{F}_q)$. Thus the objective is to compute $B$ from $r(B)$, $r(B')$, and $B'$. 

Suppose $D_1=B\setminus B'$, $D_2=B'\setminus B$, $D_3=B\cap B'$. Let $f_1= f(X;B)$ and $f_2=f(X;B')$. Then,

\begin{equation}
\begin{split}
g_1(X)&=\frac{f_1}{GCD(f_1,f_2)}=\prod\limits_{b\in D_1}(X+b),\\
g_2(X)&=\frac{f_2}{GCD(f_1,f_2)}=\prod\limits_{b\in D_2}(X+b),\\
g_3(X)&=GCD(f_1,f_2)=\prod\limits_{b\in D_3}(X+b).\\
\end{split}
\label{eqn:g}
\end{equation}

Notice that $g_1,g_2,g_3$ uniquely determine $f_1,f_2$, so they are sufficient for computing $\pi$. We next seek to compute $g_1,g_2,g_3$ from $r(B)$ and $f_2=g_2\cdot g_3$, from which $f_1=g_1\cdot g_3$ can be determined. Let $(h_1,h_2)=(X^{t-k}g_2,X^{t-k}g_1)$, where $k=\deg g_1=\deg g_2=\lvert D_1\rvert=\lvert D_2\rvert\leq t$. Then $(h_1,h_2)$ satisfy $h_1\cdot f_1=h_2\cdot f_2$. We will also prove later in \Cref{theo3} that $g_1,g_2,g_3$ can be computed from an arbitrary nonzero solution $(h_1,h_2)$ of $h_1\cdot f_1=h_2\cdot f_2$. Therefore, any nonzero solution to $h_1\cdot f_1=h_2\cdot f_2$ is sufficient for computing $\pi$. Also notice that the first $4t$ coefficients of $h_1\cdot f_1$ and $h_2\cdot f_2$ uniquely determine $r(B)$ and $r(B')$, respectively, by (\ref{eqn:ni}), if $h_1,h_2$ are known. In order to compute $g_1,g_2,g_3$, it is sufficient to find $h_1$ and $h_2$, both of degree $t$, such that the first $4t$ coefficients of $h_1\cdot f_1$ and that of $h_2\cdot f_2$ are equal, i.e., the following inequality holds,
\begin{equation}
\deg(h_1\cdot f_1-h_2\cdot f_2)<N-3t.
\label{eqn:decode}
\end{equation} 

For each $\V{c}\in \Db{F}_q^{2t}$, suppose
\begin{equation}
\V{c}=\begin{pmatrix}c_1,\cdots,c_t,-c'_1,\cdots,-c'_t\end{pmatrix}^{T},
\label{eqn:c}
\end{equation}
and define the polynomials $h_1(\V{c}),h_2(\V{c})$ of degree $t$  as follows, 
\begin{equation}
\begin{split}
h_1(\V{c})&\triangleq X^t+c_1 X^{t-1}+c_2 X^{t-2}+\cdots+c_t,\\
h_2(\V{c})&\triangleq X^t+c'_1 X^{t-1}+c'_2 X^{t-2}+\cdots+c'_t.\\
\end{split}
\label{eqn:h}
\end{equation}

Define
\begin{equation}
\begin{split}
&\V{A}=\\
&
\begin{pmatrix}
1        &0 &\cdots&0 & 1       &0 & \cdots&0\\
a_1      & 1        & \ddots&\vdots& a'_1      & 1        & \ddots&\vdots\\
\vdots      & \vdots      & \ddots& 0&\vdots      & \vdots     & \ddots&0\\
a_{t-1}      & a_{t-2}     & \cdots  &   1&a'_{t-1}      & a'_{t-2}      &  \cdots &   1\\
\vdots   & \vdots   &  \ddots &  \vdots&\vdots   & \vdots   &  \ddots &  \vdots\\
a_{4t-2} & a_{4t-3} & \cdots & a_{3t-1}&a'_{4t-2} & a'_{4t-3} & \cdots & a'_{3t-1}
\end{pmatrix}
,\\
\end{split}
\label{eqn:a}
\end{equation}
and
\begin{equation}
\V{b}=\begin{pmatrix}a'_1,\cdots,a'_{4t-1}\end{pmatrix}^{T}-\begin{pmatrix}a_1,\cdots,a_{4t-1}\end{pmatrix}^T.
\label{eqn:b}
\end{equation}

The following \Cref{algo: Decode1} describes the decoding algorithm of the code constructed in \Cref{subsec: nonsysencode}. The correctness of this algorithm is proved by \Cref{lemma8} and \Cref{theo3}.

\begin{algorithm}
\caption{Decoding Algorithm}\label{algo: Decode1}
\begin{algorithmic}[1]
\Require 
\Statex Syndrome: $\bm{\alpha}$;
\Statex Received sequence: $\pi'$;
\Ensure
\Statex Estimated codeword: $\hat{\pi}$;
\State Compute the coefficients $\LB{a'_i}_{i=1}^{4t-1}$ of $f_2$ and $B'$ from $\pi'$;
\State Compute the coefficients of $\LB{a_i}_{i=1}^{4t-1}$ of $f_1$ from $\bm{\alpha}$ by \textit{Newton's identities};
\State Compute $\V{A}$ and $\V{b}$ using (\ref{eqn:a}) and (\ref{eqn:b});
\State Find a nonzero solution $\V{c}$ to $\V{A}\V{c}=\V{b}$, $\V{c}=\begin{pmatrix}c_1,\cdots,c_{2t}\end{pmatrix}^{T}$;
\State Compute $h_1=X^t+c_1X^{t-1}+c_2X^{t-2}+\cdots+c_t$, $h_2=X^t-c_{t+1}X^{t-1}-c_{t+2}X^{t-2}-\cdots-c_{2t}$;
\State Compute $h=\gcd (h_1,h_2)$, $v_1=\frac{h_2}{h}$, $v_2=\frac{h_1}{h}$;
\State Let the set of negative roots of $v_1$ and $v_2$ be $V_1$ and $V_2$, respectively;
\State Compute $\hat{\pi}=\nu^{-1}\SB{V_1\cup( B' \setminus V_2)}$, where $\nu$ is defined in (\ref{eqn:defig});
\State \textbf{return} $\hat{\pi}$.
\end{algorithmic}
\end{algorithm}

\Cref{lemma8} presents an equivalent linear equation to find a solution that satisfies (\ref{eqn:decode}), and \Cref{theo3} shows how to compute $\pi$ from this intermediate value.

\begin{lemma}\label{lemma8} Suppose $\V{A}\in \Db{F}_q^{(4t-1)\times(2t)}$, $\V{b}\in \Db{F}_q^{4t-1}$ are defined in (\ref{eqn:a}) and (\ref{eqn:b}), respectively. Consider the following equation defined on $\Db{F}_q$:
\begin{equation}
\V{A}\V{c}=\V{b}.
\label{eqn:dec}
\end{equation}

For any vector $\V{c}\in\Db{F}_q^{2t}$, $\V{c}$ is a nonzero solution to (\ref{eqn:dec}) if and only if $(h_1(\V{c}),h_2(\V{c}))$ is a nonzero solution to (\ref{eqn:decode}).
\end{lemma}
\begin{proof} The proof is in Appendix \ref{prooflemma8}. 
\end{proof}

\begin{theo}\label{theo3} Let $\V{c}$ be an arbitrary nonzero solution to (\ref{eqn:dec}), and $h_1=h_1(\V{c})$, $h_2=h_2(\V{c})$. Denote $h,v_1,v_2$ by the following equations,
\begin{equation}
h=\text{GCD}(h_1,h_2), v_1=\frac{h_2}{h},\ v_2=\frac{h_1}{h}.
\label{eqn:GCD}
\end{equation} 

Suppose $V_1,V_2$ are the sets of the additive inverses of roots of $v_1,v_2$, respectively. Then $\pi$ can be computed from the following equation:
\begin{equation*}
\pi=\nu^{-1}\left(V_1\cup\left(B'\setminus V_2\right)\right).
\end{equation*} 
Recall $B'=\nu(\pi')$, where $\nu$ is defined in (\ref{eqn:defig}).
\end{theo}

\begin{proof}
Note that $B=\nu(\pi)$ and $\nu$ is an injection, so we only need to prove that $B=V_1\cup\left(B'\setminus V_2\right)$. From (\ref{eqn:g}), it follows that
\begin{equation*}
h_1\cdot f_1-h_2\cdot f_2=(h_1\cdot g_1-h_2\cdot g_2)\cdot g_3,
\end{equation*}
where $\deg g_3=\lvert B\cap B'\rvert \geq N-1-t$. From \Cref{lemma8}, (\ref{eqn:decode}) holds true, which means that $\deg (h_1\cdot g_1-h_2\cdot g_2)\cdot g_3=\deg (h_1\cdot f_1-h_2\cdot f_2)<N-3t$. If $h_1\cdot f_1\neq h_2\cdot f_2$, then $N-t-1<N-3t$ and thus $t=0$, $h_1=h_2=0$. Therefore for any nonzero pair of $h_1$ and $h_2$,
\begin{equation*}
h_1\cdot f_1=h_2\cdot f_2.
\end{equation*}
We know from (\ref{eqn:GCD}) that
\begin{equation*}
v_2\cdot f_1=v_1\cdot f_2,
\end{equation*}
where $\text{GCD}(v_1,v_2)=1$. Let $v_2|f_2$ and $v_1|f_1$. Then,
\[\frac{f_1}{v_1}=\frac{f_2}{v_2}=f.\]

Suppose $V_3$ is the set of the additive inverses of roots of $f$. Then $V_1\cup V_3=B$, $V_2\cup V_3=B'$, thus $B=V_1\cup V_3=V_1\cup\left(B'\setminus V_2\right)$.
\end{proof}

Note that $V_1,V_2$ computed in \Cref{theo3} are exactly identical to $D_1,D_2$ described before (\ref{eqn:g}), respectively.

\begin{exam} Suppose the sender transmits the permutation $\pi_1=(2,4,7,3,5,1,8,6,9,10)\in\Se{C}_{\bm{\alpha}}(10,2)$, where $\bm{\alpha}=(16,0,86,44,61,9,49)$, and the receiver recives $\pi{'}=(8,6,9,10,5,1,2,4,7,3)\in \Db{S}_{10}$. In the encoding scheme, $q=97>10^2-10$, and for all $i,j\in\left[10\right]$, $i\neq j$,
\begin{equation*}
\upsilon(i,j)=10(i-1)+j-1.
\end{equation*}

The receiver applies \emph{Newton's identities}\cite{zeilberger1984combinatorial} to compute $r(B)=(16,31,0,  42 ,   54  ,  94 ,59)$ from $\bm{\alpha}$, and then derives $r(B')=( 80  ,  64  ,  83  ,  10   , 72   , 22,26)$ from $B'=\nu(\pi{'})=\{75,58,89,94,40,1,13,36,62\}$. Then 

\begin{equation}
\begin{split}
\V{A}&=\begin{pmatrix}
1   & 16   &   31   &    0   &   42  &    54   &   94\\
     0   &    1    &  16   &   31   &    0    &  42   &   54\\
     1    &  80   &   64   &   83   &   10   &   72   &   22\\
     0    &   1    &  80    &  64   &   83    &    10   &   72
\end{pmatrix}^T,\\
\V{b}&=\begin{pmatrix}64   & 33 &   83 &   65 &   18 &   25 &   64\end{pmatrix}^T.
\end{split}
\end{equation}

Notice that $\V{c}=\begin{pmatrix}95,94,66,26\end{pmatrix}$ is a solution to $\V{A}\V{c}=\V{b}$. Therefore $h_1=X^2+95X+94=(X+1)(X+94)$, $h_2=X^2+31X+71=(X+24)(X+7)$. The receiver then knows that $V_1=\{24,7\}$, $V_2=\{1,94\}$. Therefore $\nu(\pi)=B=V_1\cup(B'\setminus V_2)=\{13,36,62,24,40,7,75,58,89\}$. It follows that $A(\pi)=\{(2,4),(4,7),(7,3),(3,5),(5,1),(1,8),(8,6),(6,9),(9,10)\}$. From the definition of the characteristic set in \Cref{defi4}, the receiver is able to decode $\pi$ from $A(\pi)$ as $\hat{\pi}=(2,4,7,3,5,1,8,6,9,10)$.

\end{exam}

\section{Systematic Permutation Codes in the Generalized Cayley Metric}
\label{sec: syscodes}
In this section, we discuss systematic permutation codes. Specifically, in \Cref{subsec: sysencode}, we present an explicit coding scheme for systematic permutation codes in the generalized Cayley metric, and in \Cref{subsec: sysdecode}, we provide the decoding scheme for this construction. We refine our construction to ensure order-optimality, which we then discuss in \Cref{subsec: optimalsyscodes}.



\subsection{Encoding Scheme}
\label{subsec: sysencode}
Let messages be permutations in $\Db{S}_N$. In systematic permutation codes, the codewords are permutations of length $N+M$. We derive each codeword $\sigma\in\Db{S}_{N+M}$ from a message $\pi\in\Db{S}_N$ by sequentially inserting components $N+1,N+2,\cdots,N+M$ into $\pi$, in the positions specified by a sequence $S=(s_1,s_2,\cdots,s_M)$, where $S$ is determined by the syndrome $\alpha^{(q,2t)}(\pi)$ defined in (\ref{eqn:defialpha1}) and (\ref{alphadefi}). Our key result is established in \Cref{theo4}, where we present the construction of systematic permutation codes. We start the discussion by presenting a collection of definitions and lemmas to support our main result.

\begin{Defi}\label{defi5} For any permutation $\pi\in\Db{S}_N$ and the integer $i\in\Db{N}$, where $1\leq s\leq N$, let $E(\pi,s)$ be a permutation in $\Db{S}_{N+1}$ derived by inserting the element $N+1$ after the element $s$ in $\pi$, i.e., 
\begin{equation*}
E(\pi,s)\triangleq \SB{\pi(1),\cdots,\pi(k),N+1,\pi(k+1),\cdots,\pi(N)},
\end{equation*}
where $k=\pi^{-1}(s)$. We call $E(\pi,s)$ the \textbf{extension} of $\pi$ on the \textbf{extension point} $s$.

Consider a sequence $S=\SB{s_1,s_2,\cdots,s_{M}}$, where $s_m\in\MB{N}$ for all $1\leq m\leq M$. The \textbf{extension} $E(\pi,S)$ of $\pi$ on the \textbf{extension sequence} $S$ is a permutation in $\Db{S}_{N+M}$ derived from inserting the elements $N+1,\cdots,N+M$ sequentially after the elements $s_1,\cdots,s_M$ in $\pi$, i.e.,
\begin{equation*}
E(\pi,S)\triangleq E(E(\cdots E(E(\pi,s_1),s_2)\cdots,s_{M-1}),s_M).
\end{equation*}

\end{Defi}

Note that in \Cref{defi5}, the elements $s_1,\cdots,s_M$ in the extension sequence $S$ are not necessarily distinct. If different symbols are sequentially inserted after the same element, then they are all placed right after this element in descending order, as shown in \Cref{exam4}.

\begin{exam}\label{exam4} Suppose $\pi=(1,4,5,7,6,2,3)$, $I=(4,1,2,2)$, then
\begin{equation*}
\begin{split}
E(\pi,I)&=\SB{1,9,4,8,5,7,6,2,11,10,3}.\\
\end{split}
\end{equation*}
\end{exam}

Based on the definition of the extensions, \Cref{algo: Encode2} describes the major steps of our encoding scheme. The correctness of this scheme is proved later by \Cref{lemma9} and \Cref{theo4}.

\begin{algorithm}
\caption{Encoding Scheme}\label{algo: Encode2}
\begin{algorithmic}[1]
\Require 
\Statex Information sequence: $\pi\in \Db{S}_N$;
\Statex Number of additional symbols: $K$;
\Statex Minimum block permutation distance: $2t+1$;
\Ensure
\Statex Codeword: $\sigma$ ($\sigma\in\Db{S}_{N+K}$);
\State Compute the syndrome $\bm{\alpha}=\alpha^{(q,2t)}(\pi)$ of $\pi$, which is defined in (\ref{eqn:defialpha1});
\State Compute the extension sequence $S=\varphi(\bm{\alpha})$, where $\varphi$ is a function such that the image of $\varphi$ is a $t$-auxiliary set of length $K$ in the range $\MB{N}$, as defined in \Cref{defi9};
\State Compute $\sigma=E(\pi,S)$, according to \Cref{defi5};
\State \textbf{return} $\sigma$.
\end{algorithmic}
\end{algorithm}

\Cref{defi6} presents the notion of the \emph{jump points} of the extensions of two permutations. Then \Cref{lemma9} states that the block permutation distance between two extensions is strictly larger than that of their original permutations if and only if the extension point of one of them is a jump point. Based on this result, we further introduce the notion of \emph{jump index} and \emph{jump set} in \Cref{defi7}. As shown in \Cref{rem5}, the block permutation distance of two permutations in $\Db{S}_N$ is lower bounded by the sum of that of their extensions and the cardinality of the jump set. 

\begin{Defi}\label{defi6} Let $\pi_1,\pi_2\in \Db{S}_N$, $s_1,s_2\in\MB{N}$. We note that for any $k\in\MB{N}$, $\pi_{i}(k)$ refers to the $k$-th element of $\pi_i$, $i\in\LB{1,2}$. Suppose $E(\pi_1,s_1)$, $E(\pi_2,s_2)$ are two arbitrary extensions of $\pi_1$ and $\pi_2$, respectively, where $\pi_1,\pi_2\in\Db{S}_N$, $\pi_{1}(k_1)=s_1$ and $\pi_{2}(k_2)=s_2$. Then $s_1$ is called a \textbf{jump point} of $E(\pi_1,s_1)$ with respect to $E(\pi_2,s_2)$, if $s_1\neq s_2$ and at least one of the following conditions is satisfied:

\begin{enumerate}
  \item $k_1=N$ or $k_2=N$;
  \item $k_1,k_2<N$, and $\pi_{1}(k_1+1)\neq \pi_{2}(k_2+1)$.
\end{enumerate}
\end{Defi}

\begin{lemma}\label{lemma9} Let $\pi_1,\pi_2\in \Db{S}_N$, $s_1,s_2\in\MB{N}$. For any two extensions $E(\pi_1,s_1)$ and $E(\pi_2,s_2)$, if $s_1$ is a jump point of $E(\pi_1,s_1)$ with respect to $E(\pi_2,s_2)$, then
\begin{equation}
d_B(E(\pi_1,s_1),E(\pi_2,s_2))>d_B(\pi_1,\pi_2),
\end{equation}
else
\begin{equation}
d_B(E(\pi_1,s_1),E(\pi_2,s_2))=d_B(\pi_1,\pi_2).
\end{equation}

\end{lemma} 

\begin{proof} The proof is in Appendix \ref{prooflemma9}.
\end{proof}

In the following \Cref{exam5}, we provide examples of jump points that satisfy the two conditions indicated in \Cref{defi6}. We also provide an example of an extension point that is not a jump point.

\begin{exam}\label{exam5} Suppose $\pi=(1,5,7,2,3,6,4)$, $\pi'=(2,3,1,5,7,6,4)$, $s_1=4$, $s_1'=5$, $s_2=5$, $s_2'=6$, $s_3=3$, $s_3'=7$. Then,
\begin{equation*}
\begin{split}
\sigma_1=E(\pi,s_1)&=\SB{1,5,7,2,3,6,4,8},\\
\sigma_1'=E(\pi',s'_1)&=\SB{2,3,1,5,8,7,6,4},\\
\sigma_2=E(\pi,s_2)&=\SB{1,5,8,7,2,3,6,4},\\
\sigma_2'=E(\pi',s'_2)&=\SB{2,3,1,5,7,6,8,4},\\
\sigma_3=E(\pi,s_3)&=\SB{1,5,7,2,3,8,6,4},\\
\sigma_3'=E(\pi',s'_3)&=\SB{2,3,1,5,7,8,6,4}.\\
\end{split}
\end{equation*}

Given that $d_B(\pi,\pi')=2$, we observe that 
\begin{equation*}
\begin{split}
&d_B(\sigma_1,\sigma'_1)=4>d_B(\pi,\pi'),\text{ and } s_1\text{ is a jump point;}\\
&d_B(\sigma_2,\sigma'_2)=5>d_B(\pi,\pi'),\text{ and }  s_2\text{ is a jump point;}\\
&d_B(\sigma_3,\sigma'_3)=2=d_B(\pi,\pi'),\text{ and }  s_3\text{ is not a jump point}.\\
\end{split}
\end{equation*}

Notice that $s_1$ is a jump point that satisfies the first condition in \Cref{defi6}, and $s_2$ is a jump point that satisfies the second condition. This example is consistent with \Cref{lemma9}.
\end{exam}

We know from \Cref{lemma9} that the block permutation distance between the resulting codewords cannot be smaller than that of their original messages. Recall that \Cref{theo2} indicates that permutations with the same syndrome result in codewords having the block permutation distance of at least $2t+1$. Therefore, it suffices to show that the permutations with different syndromes are mapped to codewords that are sufficiently far apart under the block permutation distance; \Cref{lemma10} establishes a property that ensures that this condition is satisfied. We then use this result in \Cref{theo4} to present the construction of systematic permutation codes.

\begin{Defi}\label{defi7} Let $\pi_1,\pi_2\in \Db{S}_N$, $s_1,s_2\in\MB{N}$. Suppose $E(\pi_1,S_1)$ and $E(\pi_2,S_2)$ are extensions of $\pi_1$ and $\pi_2$ on extension sequences $S_1$ and $S_2$, respectively, where $\pi_1,\pi_2\in\Db{S}_N$, $S_1=\SB{s_{1,1},s_{1,2},\cdots,s_{1,M}}$ and $S_2=\SB{s_{2,1},s_{2,2},\cdots,s_{2,M}}$. Then, for any $m\in\MB{M}$, $m$ is called a \textbf{jump index} of $E(\pi_1,S_1)$ and $E(\pi_2,S_2)$ if $s_{1,m}$ is a jump point of $E(E(\pi_1,J_{1,m-1}),s_{1,m})$ with respect to $E(E(\pi_2,J_{2,m-1}),s_{2,m})$, where $J_{1,m-1}=\SB{s_{1,1},s_{1,2},\cdots,s_{1,m-1}}$, $J_{2,m-1}=\SB{s_{2,1},s_{2,2},\cdots,s_{2,m-1}}$. Define the \textbf{jump set} $F(\pi_1,\pi_2,S_1,S_2)$ as the set of all jump indices of $E(\pi_1,S_1)$ and $E(\pi_2,S_2)$. 
\end{Defi}

\begin{rem}\label{rem5} Let $\pi_1,\pi_2\in \Db{S}_N$, $s_1,s_2\in\MB{N}$.
For any extensions $E(\pi_1,S_1)$, $E(\pi_2,S_2)$ of $\pi_1$, $\pi_2$ on extension sequences $S_1$, $S_2$, respectively, it is obvious from \Cref{defi7} and \Cref{lemma9} that 
\begin{equation}
d_B(E(\pi_1,S_1),E(\pi_2,S_2))\geq d_B(\pi_1,\pi_2)+\lvert F(\pi_1,\pi_2,S_1,S_2) \rvert.
\label{extensiondistance}
\end{equation}
Here $F(\pi_1,\pi_2,S_1,S_2)$ is the jump set defined in \Cref{defi7}.
\end{rem}

In the following \Cref{exam6}, we provide an example of how to identify the jump indices and compute the jump set. This example satisfies inequality (\ref{extensiondistance}).

\begin{exam}\label{exam6} Continuing with the values of $\pi$, $\pi'$ specified in \Cref{exam5}, let $S=\SB{4,6,7}$ and $S'=\SB{5,6,5}$. Then,
\begin{equation*}
\begin{split}
\sigma_0=\pi&=\SB{1,5,7,2,3,6,4},\\
\sigma_0'=\pi'&=\SB{2,3,1,5,7,6,4},\\
\sigma_1=E(\sigma_0,s_1)&=\SB{1,5,7,2,3,6,4,8},\\
\sigma_1'=E(\sigma_0',s'_1)&=\SB{2,3,1,5,8,7,6,4},\\
\sigma_2=E(\sigma_1,s_2)&=\SB{1,5,7,2,3,6,9,4,8},\\
\sigma_2'=E(\sigma_1',s'_2)&=\SB{2,3,1,5,8,7,6,9,4},\\
\sigma_3=E(\sigma_2,s_3)&=\SB{1,5,7,10,2,3,6,9,4,8},\\
\sigma_3'=E(\sigma_2',s'_3)&=\SB{2,3,1,5,10,8,7,6,9,4}.\\
\end{split}
\end{equation*}

It follows immediately that
\begin{equation*}
\begin{split}
d_B(\sigma_0,\sigma'_0)&=2,\\
d_B(\sigma_1,\sigma'_1)&=4>d_B(\sigma_0,\sigma'_0),\text{ and }  1\text{ is a jump index};\\
d_B(\sigma_2,\sigma'_2)&=4=d_B(\sigma_1,\sigma'_1),\text{ and }  2\text{ is not a jump index};\\
d_B(\sigma_3,\sigma'_3)&=5>d_B(\sigma_2,\sigma'_2),\text{ and }  3\text{ is a jump index}.\\
\end{split}
\end{equation*}

According to \Cref{defi7}, $F(\pi,\pi',S,S')=\LB{1,3}$. Moreover, $d_B(\sigma_3,\sigma_3')=5>4=d_B(\pi,\pi')+\lvert F(\pi,\pi',S,S')\rvert$, which is in accordance with equation (\ref{extensiondistance}). 
\end{exam}

Next we prove in \Cref{lemma10} that the right hand side of equation (\ref{extensiondistance}) can be lower bounded by the cardinality of the so-called \emph{Hamming set}. The Hamming set of $S_1$ with respect to $S_2$ is defined in the following \Cref{defi8}. Based on this result, we present a construction of systematic $t$-block permutation codes in \Cref{theo4} with the help of a so-called \emph{$t$-auxiliary set} that is defined in \Cref{defi9}.

\begin{Defi}\label{defi8} For any sequences $\V{v}_1$, $\V{v}_2$ of integers with length $M$, where $\V{v}_1=\SB{v_{1,1},v_{1,2},\cdots,v_{1,M}}$ and $\V{v}_2=\SB{v_{2,1},v_{2,2},\cdots,v_{2,M}}$, define the \textbf{Hamming set} of $\V{v}_1$ with respect to $\V{v}_2$ as follows,
\begin{equation}
H(\V{v}_1,\V{v}_2)\triangleq\LB{v_{1,m}|v_{1,m}\neq v_{2,m}, m\in\MB{M}}.
\end{equation}
\end{Defi}

We note that $d_H$ refers to the Hamming distance throughout this paper.

\begin{rem}\label{rem6} It is obvious that $d_H(\V{v}_1,\V{v}_2)\geq \Nm{H(\V{v}_1,\V{v}_2)}$. Additionally, for any three sequences $\V{v}_1,\V{v}_2,\V{v}_3$ of integers, the following triangle inequality holds true:
\begin{equation} 
\Nm{H(\V{v}_1,\V{v}_3)}\leq \Nm{H(\V{v}_1,\V{v}_2)}+\Nm{H(\V{v}_2,\V{v}_3)}.
\label{eqn: remhammingset}
\end{equation}
\end{rem}

\begin{lemma}\label{lemma10} Let $\pi_1,\pi_2\in \Db{S}_N$, $s_1,s_2\in\MB{N}$. For any extensions $E(\pi_1,S_1)$, $E(\pi_2,S_2)$ of $\pi_1$, $\pi_2$ on extension sequences $S_1$, $S_2$, respectively, it follows that
\begin{equation}
d_B(E(\pi_1,S_1),E(\pi_2,S_2))\geq \lvert H(S_1,S_2)\rvert.
\end{equation}

\end{lemma}

\begin{proof} The proof is in Appendix \ref{prooflemma10}.
\end{proof}

\begin{exam}\label{exam7} Continuing on with the numerical values of $\pi,\pi',S,S'$ as in \Cref{exam6}, we conclude that, $H(S,S')=\LB{4,7}$, $m(4)=1$, $m(7)=3$. Then it follows that $d_B(\sigma,\sigma')=5>2=\Nm{H(S,S')}$, which is in accordance with the above \Cref{lemma10}.
\end{exam}

\begin{Defi}\label{defi9} Consider a set $\Se{A}(N,K,t)\subset \MB{N}^K$. We call $\Se{A}(N,K,t)$ a $\bm{t}$\textbf{-\emph{auxiliary set}} of length $K$ in range $\MB{N}$ if for any $\V{c}_1,\V{c}_2\in\Se{A}(N,K,t)$, $\V{c}_1\neq \V{c}_2$, $\Nm{H(\V{c}_1,\V{c}_2)}\geq 2t+1$ holds.
\end{Defi}

\begin{theo}\label{theo4} For any $t$-auxiliary set $\Se{A}(N,K,t)$ with cardinality that is no less than $q^{4t-1}$, suppose $\varphi:\ \alpha^{(q,2t)}(\Db{S}_N)\to \Se{A}(N,K,t)$ is an arbitrary injection, where $q$ is a prime number such that $N^2-N<q<2(N^2-N)$ and the syndrome $\alpha^{(q,2t)}$ is defined in (\ref{eqn:defialpha1}) and (\ref{alphadefi}). Then, the set $\Se{C}_{B}^{\Tx{sys}}(N,K,t)=\LB{E(\pi,\varphi\comp\alpha^{(q,2t)}(\pi))|\pi\in\Db{S}_N}$ is a systematic $t$-block permutation code.
\end{theo}

\begin{proof} It is clear by the choice of $E(\pi,S)$ that $\Se{C}_{B}^{\Tx{sys}}(N,K,t)$ is systematic. For any two messages $\pi_1,\pi_2\in\Db{S}_N$, denote their corresponding codewords by $\sigma_1=E(\pi_1,\varphi\comp\alpha^{(q,2t)}(\pi_1))$ and $\sigma_2=E(\pi_2,\varphi\comp\alpha^{(q,2t)}(\pi_2))$, respectively. Suppose $\bm{\alpha}_1=\alpha^{(q,2t)}(\pi_1)$, $\bm{\alpha}_2=\alpha^{(q,2t)}(\pi_2)$, $S_1=\varphi(\bm{\alpha}_1)$ and $S_2=\varphi(\bm{\alpha}_2)$. Then $\sigma_1=E(\pi_1,S_1)$, $\sigma_2=E(\pi_2,S_2)$. Consider the following two cases:

\begin{enumerate}
\item $\bm{\alpha}_1=\bm{\alpha}_2$. According to \Cref{theo2}, $d_B(\pi_1,\pi_2)>2t$ in this case. Then \Cref{lemma9} implies that $d_B(\sigma_1,\sigma_2)\geq d_B(\pi_1,\pi_2)\geq 2t+1$.
\item $\bm{\alpha}_1\neq \bm{\alpha}_2$. In this case, $S_1,S_2 \in\Se{A}(N,K,t)$ and $S_1\neq S_2$. Then from \Cref{defi9}, $\Nm{H(S_1,S_2)}\geq 2t+1$. Therefore, from \Cref{lemma10}, $d_B(\sigma_1,\sigma_2)\geq \Nm{H(S_1,S_2)}\geq 2t+1$.
\end{enumerate}

From the above discussion, $d_B(\sigma_1,\sigma_2)\geq 2t+1$ is aways true, which means that $\Se{C}_{B}^{\Tx{sys}}(N,K,t)$ is indeed a systematic $t$-block permutation code.
\end{proof}

\subsection{Decoding Scheme}
\label{subsec: sysdecode}
Based on the construction and the notation in \Cref{theo4}, suppose the sender sends a codeword $\sigma=E(\pi,\varphi\comp\alpha^{(q,2t)}(\pi))$ through a noisy channel and the receiver receives a noisy version $\sigma'$, where $d_B(\sigma,\sigma')\leq t$. 

 In this section, we prove in the forthcoming \Cref{lemma11} that the extension sequence $S$ of the codeword $E(\pi,S)$ is decodable given that $d_B(\sigma,\sigma')\leq t$, from which the syndrome, defined in (\ref{eqn:defialpha1}) and (\ref{alphadefi}), of the transmitted information $\pi$ can be derived.

For convenience, we introduce the following definition of the \emph{truncation} and use it throughout this subsection.

\begin{Defi}\label{defi10} For any permutation $\sigma\in\Db{S}_{N+1}$ and an integer $u\in\MB{N+1}$, denote $T(\sigma,u)$ to be the sequence derived by removing the element $u$ from  $\sigma$, i.e.,
\begin{equation}
\begin{split}
T(\sigma,u)\triangleq \SB{\sigma(1),\sigma(2),\cdots,\sigma(k-1),\sigma(k+1),\cdots,\sigma(N)},
\end{split}
\end{equation}
where $k=\sigma^{-1}(u)$.

Then, for any permutation $\sigma\in\Db{S}_{N+M}$ and a set $U\subset \MB{N+M}$, denote the \textbf{truncation} $T(\sigma,U)$ of $\sigma$ on set $U$ to be the sequence derived by removing the elements contained in $U=\LB{u_1,u_2,\cdots,u_{|U|}}$ from $\sigma$, i.e.,
\begin{equation}
\begin{split}
T(\sigma,U)\triangleq T(T(\cdots T(T(\sigma,u_1),u_2)\cdots,u_{\Nm{U}-1}),u_{\Nm{U}}).
\end{split}
\end{equation}
\end{Defi}

Note that in \Cref{defi10}, the ordering of $u_1,\cdots,u_{\Nm{U}}$ has no impact on the value of $T(\sigma,U)$. The following is an example of the truncation of a permutation.

\begin{exam}\label{exam8} Suppose $\sigma=\SB{1,4,5,2,3,9,8,6,7}$, $U=\LB{4,5,9}$, then
\begin{equation*}
\begin{split}
T(\sigma,U)=\SB{1,2,3,8,6,7}.
\end{split}
\end{equation*}
\end{exam}

The following \Cref{algo: Decode2} describes the decoding algorithm of the code constructed in \Cref{theo4}. The correctness of this algorithm is proved by \Cref{lemma11}.

\begin{algorithm}
\caption{Decoding Algorithm}\label{algo: Decode2}
\begin{algorithmic}[1]
\Require 
\Statex Received sequence: $\sigma'$;
\Statex Number of additional symbols: $K$;
\Statex Minimum block permutation distance: $2t+1$;
\Ensure
\Statex Estimated information sequence: $\hat{\pi}$;
\State Compute $\pi'=T(\sigma',\LB{N+1,\cdots,N+K})$, according to \Cref{defi10};
\State Find $S'$ such that $\sigma'=E(\pi',S')$, where $E(\pi,S)$ is defined in \Cref{defi5};
\State Find $\hat{S}\in \rm{Img}(\varphi)$ such that $H(\hat{S},S')\leq t$, where $H$ is defined in \Cref{defi8}, and $\varphi$ is specified in \Cref{theo4};
\State Compute $\hat{\bm{\alpha}}=\varphi^{-1}(\hat{S})$;
\State Let $\hat{\bm{\alpha}},\pi'$ be the inputs of \Cref{algo: Decode1} and obtain $\hat{\pi}$;
\State \textbf{return} $\hat{\pi}$.
\end{algorithmic}
\end{algorithm}

Our decoding scheme has two major steps. Recall that $\alpha^{(q,2t)}$ is defined in (\ref{eqn:defialpha1}) and (\ref{alphadefi}) as the syndrome of $\pi$. The first step is to derive the syndrome $\hat{\bm{\alpha}}=\alpha^{(q,2t)}(\pi)$ of $\pi=T(\sigma,\LB{N+1,\cdots,N+K})$, from the received permutation $\sigma'$. The second step is to apply \Cref{algo: Decode1} to the pair of inputs, the syndrome $\hat{\bm{\alpha}}$ and the subsequence $\pi'=T(\sigma',\LB{N+1,\cdots,N+K})$, and compute $\pi$.

Note that it is sufficient to compute the sequence $S$ in order to derive the syndrome $\hat{\bm{\alpha}}$. \Cref{lemma11} proves the sufficiency of obtaining the sequence $S$ from $S'$, where $S$ is the extension sequence of $\pi$ in $\sigma$, by showing that the cardinality of the Hamming set $H(S,S')$ does not exceed $t$, provided that $d_B(\sigma,\sigma')\leq t$. Therefore, from (\ref{eqn: remhammingset}) and \Cref{defi9}, we are able to obtain an estimate $\hat{S}$ of $S$ from $S'$ since each $t$-auxiliary set $\Se{A}(N,K,t)$ has the property that the cardinalities of Hamming sets constructed from its pairwise distinct elements are at least $2t+1$. The syndrome $\hat{\bm{\alpha}}$ is then uniquely derived from $\hat{S}$.

\begin{lemma}\label{lemma11} Consider an arbitrary $\sigma\in\Se{C}=\LB{E(\pi,\varphi\comp\alpha^{(q,2t)}(\pi))|\pi\in\Db{S}_N}$, for $\Se{C}$ defined in \Cref{theo4} (then $\sigma\in\Db{S}_{N+K}$). Suppose there is a $\sigma'$ such that $d_B(\sigma,\sigma')\leq t$. Let $S=\varphi\comp\alpha^{(q,2t)}(\pi)$ and $\pi'=T(\sigma',\MB{N+1:N+K})$. Suppose $\sigma'$ is the extension of $\pi'$ on the extension sequence $S'$, i.e., $\sigma'=E(\pi',S')$. Then,
\begin{equation}
H(S,S')\leq t.
\end{equation}
\end{lemma}

\begin{proof} Suppose $S=\SB{s_1,s_2,\cdots,s_{K}}$, $S'=\SB{s'_1,s'_2,\cdots,s'_K}$. Then, according to \Cref{theo4}, $S\in\Se{A}(N,K,t)$. Let $\Se{M}=\LB{m|s_m\neq s_m',1\leq m\leq K}$. For all $m\in\Se{M}$, it follows from \Cref{defi6} that there exist subsequences of $\sigma,\sigma'$: $\V{p}_m=(s_m,n_{k(m)},n_{k(m)-1},\cdots,n_1,N+m)$ and $\V{p}'_m=(s'_m,n_{k'(m)}',n_{k'(m)-1}',\cdots,n_1',N+m)$, where $k(m),k'(m)\in\MB{K}$, $n_1,n_2,\cdots,n_{k(m)},n_1',n_2',\cdots,n_{k(m)'}'\in\MB{N+1:N+K}$. Note that $s_m\neq s'_m$, which means that $(s_m,n_{k(m)},n_{k(m)-1},\cdots,n_1)\neq (s'_m,n_{k'(m)}',n_{k'(m)-1}',\cdots,n_1')$. Let 
\begin{equation*}
i(m)=\min\limits_{\substack{1\leq i\leq \min \LB{k(m),k'(m)}\\n_i\neq n'_i}} i.
\end{equation*} 
Then $n_{i(m)}\neq n'_{i(m)}$ and $n_{i(m)-1}=n'_{i(m)-1}$, where we let $n_0=n_0'=N+m$ if $i(m)=1$.

Recall the notion of characteristic sets in \Cref{defi3}. We know that $(n_{i(m)},n_{i(m)-1})\in A(\sigma)$ and $(n_{i(m)}',n_{i(m)-1}')\in A(\sigma')$. These two conditions along with the fact that $n_{i(m)}\neq n'_{i(m)}$ and $n_{i(m)-1}=n'_{i(m)-1}$ imply that $(n_{i(m)},n_{i(m)-1})\in \SB{A(\sigma)\setminus A(\sigma')}$ for all $m\in\Se{M}$. Notice that for all $s_m\in\LB{s_m:\ m\in\Se{M}}=H(S,S')$, the associated subsequences $\V{p}_{m}$ start with different $s_m$ and they do not overlap, which indicates that the pairs $(n_{i(m)},n_{i(m)-1})$ are distinct. Then $\Nm{A(\sigma)\setminus A(\sigma')}\geq \Nm{H(S,S')}$, which is equivalent to $H(S,S')\leq d_B(\sigma,\sigma')\leq t$.
\end{proof}

From \Cref{lemma11}, the receiver first computes $\pi'=T(\sigma',\LB{N+1,\cdots,N+K})$ and derives the extension sequence $S'$ such that $\sigma'=E(\pi',S')$. Then, the receiver decodes $\hat{S}=\varphi\comp\alpha^{(q,2t)}(\pi)\in\Se{A}(N,K,t)$ from $S'$ such that $\Nm{H(S',\hat{S})}\leq t$ and derives $\hat{\bm{\alpha}}$ from $\hat{S}$. From \Cref{lemma9}, $d_B(\pi,\pi')\leq d_B(\sigma,\sigma')\leq t$ follows. Then, the receiver can apply \Cref{algo: Decode1} to compute $\hat{\pi}$ from $\pi'$ and $\hat{\bm{\alpha}}$ reliably. The decoding scheme for the systematic $t$-block permutation code $\Se{C}$ constructed in \Cref{theo4} is then complete.

\subsection{Order-optimal Systematic $t$-Block Permutation Codes}
\label{subsec: optimalsyscodes}
\Cref{theo4} presents the construction of systematic $t$-block permutation codes with $K$ redundant symbols based on a $t$-auxiliary set $\Se{A}(N,K,t)$. When $N$ is sufficiently large and $K$ is relatively small compared to $N$, the code rate is $1-\Theta(\frac{K}{N})$, which is not necessarily order-optimal.
In this section, based on the upcoming \Cref{lemma14} and \Cref{theo6}, we provide an explicit construction of a $t$-auxiliary set of length $K=56t$ in \Cref{theo7}, from which we are able to explicitly construct an order-optimal permutation code by \Cref{theo4}.

\begin{lemma}\label{lemma14} For all $k,N\in\Db{N}^{*}$, $k>3$, $N>k^2$, consider an arbitrary subset $Y\subset \MB{k}$, where $\Nm{Y}=M<k$, $Y=\LB{i_1,i_2,\cdots,i_M}$, then 
\begin{equation}
\Tx{LCM}\SB{N+i_1,N+i_2,\cdots,N+i_M}>N^{M-\frac{k}{2}}.
\end{equation}
\end{lemma}

\begin{proof} The proof is in Appendix \ref{prooflemma14}.
\end{proof}

\begin{theo}\label{theo6} For all $N,k,d\in\Db{N}^{*}$, $N>k^2$, $k>3$, define a function $\beta^{(q,d,k)}:\ \Db{F}_q^{d}\to \MB{N+1}\times\MB{N+2}\times\cdots\times\MB{N+k}$ as follows:
\begin{equation}
\begin{split}
\beta^{(q,d,k)}(\bm{x})=&\SB{\beta_1^{(q,d,k)}(\bm{x}),\beta_2^{(q,d,k)}(\bm{x}),\cdots,\beta_{k}^{(q,d,k)}(\bm{x})}\\
\triangleq &(\gamma(\bm{x})\bmod (N+1),\gamma(\bm{x})\bmod(N+2),\\
&\cdots,\gamma(\bm{x})\bmod(N+k)),\\
\end{split}
\label{eqn:beta}
\end{equation}
where $\bm{x}=(x_{1},x_{2},\cdots, x_{d})\in\Db{F}_q^d$, $\gamma(\bm{x})\triangleq\Sma{i=1}{d} x_i q^{i-1}$. Then $\forall$ $\bm{x}_1,\bm{x}_2\in \Db{F}_q^{d}$, $\bm{x}_1\neq \bm{x}_2$,

\begin{equation}
d_H(\beta^{(q,d,k)}(\bm{x}_1),\beta^{(q,d,k)}(\bm{x}_2))>\frac{k}{2}-d(2+\log_N 2).
\end{equation}
\end{theo}

\begin{proof} For arbitrary $\bm{x}_1,\bm{x}_2\in \Db{F}_q^{d}$, $\bm{x}_1\neq \bm{x}_2$, let $\beta^{(q,d,k)}(\bm{x}_1)=(\beta_{1,1},\beta_{1,2},\cdots,\beta_{1,k})$, $\beta^{(q,d,k)}(\bm{x}_2)=(\beta_{2,1},\beta_{2,2},\cdots,\beta_{2,k})$. Let $Z=\LB{i:\ \beta_{1,i}=\beta_{2,i},\ 1\leq i\leq d}$, then $d_H(\beta^{(q,d,k)}(\bm{x}_1),\beta^{(q,d,k)}(\bm{x}_2))=k-\Nm{Z}=k-M$, where $M=\Nm{Z}$. 

Suppose $Z=\LB{i_1,i_2,\cdots,i_{M}}$. Let $\gamma_1=\gamma(\bm{x}_1)$, $\gamma_2=\gamma(\bm{x}_2)$. According to the definition of $\beta^{(q,d,k)}$ in (\ref{eqn:beta}),

\begin{eqnarray*}
\left\{
\begin{array}{cccl}
\gamma_1 &\equiv &\gamma_2 &\mod(N+i_1)\\
\gamma_1&\equiv &\gamma_2 &\mod(N+i_2)\\
 &\vdots &&\\
\gamma_1 &\equiv &\gamma_2 &\mod(N+i_{M}).
\end{array}
\right.
\end{eqnarray*}
Then,
\begin{equation*}
\gamma_1\equiv \gamma_2\mod \Tx{LCM}\SB{N+i_1,N+i_2,\cdots,N+i_{M}}.
\end{equation*}
Given that $\bm{x}_1,\bm{x}_2\in \Db{F}_q^{d}$, $\bm{x}_1\neq \bm{x}_2$, then $\gamma_1\neq \gamma_2$. From \Cref{lemma14}, it follows that
\begin{equation}\label{eqn:60}
\begin{split}
\Nm{\gamma_1-\gamma_2}\geq \Tx{LCM}\SB{N+i_1,N+i_2,\cdots,N+i_{M}}>N^{M-\frac{k}{2}}.
\end{split}
\end{equation}
Moreover, the condition $\bm{x}_1,\bm{x}_2\in \Db{F}_q^{d}$, $\bm{x}_1\neq \bm{x}_2$ implies that $0\leq \gamma_1,\gamma_2<q^d$ and $\gamma_1\neq \gamma_2$. Therefore,
\begin{equation}\label{eqn:61}
\begin{split}
\Nm{\gamma_1-\gamma_2}<q^{d}.
\end{split}
\end{equation}

According to (\ref{eqn:60}) and (\ref{eqn:61}), $N^{M-\frac{k}{2}}<\Nm{\gamma_1-\gamma_2}<q^{d}<(2N^2)^{d}$ is true, which means that $M-\frac{k}{2}<d(2+\log_N 2)$. Therefore $M<\frac{k}{2}+d(2+\log_N 2)$, and then
\begin{equation*}
\begin{split}
d_H(\bm{\beta}_1,\bm{\beta}_2)&=k-M>k-(\frac{k}{2}+d(2+\log_N 2))\\
&=\frac{k}{2}-d(2+\log_N 2).
\end{split}
\end{equation*}

The theorem is proved.
\end{proof}

\begin{exam}\label{exam9} Let $k=7$, $N=50$, $d=1$, $q=2503$, $\bm{x}_1=(280)$, $\bm{x}_2=(1008)$, then $\gamma_1=280$, $\gamma_2=1008$, and 
\begin{equation*}
\begin{split}
\bm{\beta}_1=&(280\bmod 51,280\bmod 52,\cdots,280\bmod 57)\\
=&(25,20,15,10,5,0,52),\\
\bm{\beta}_2=&(1008\bmod 51,1008\bmod 52,\cdots,1008\bmod 57)\\
=&(39,20,1,36,18,0,39).\\
\end{split}
\end{equation*}

Then $d_H(\bm{\beta}_1,\bm{\beta}_2)=5>\frac{k}{2}-d(2+\log_N 2)$, which is in accordance with \Cref{theo6}.
\end{exam}

Based on \Cref{theo6}, we provide an explicit construction of a $t$-auxiliary set $\Se{A}(N,56t,t)$ in the following \Cref{theo7}.

\begin{theo}\label{theo7} For all $N,k,t\in\Db{N}^{*}$, $k\geq 28t$, $k<\Floor{\sqrt{N}-\frac{1}{2}}$. Suppose $\Db{F}_q^{4t-1}=\LB{\bm{x}_1,\bm{x}_2,\cdots,\bm{x}_{q^{4t-1}}}$, where $q$ is a prime number such that $N^2-N<q<2N^2-2N$. For any $s\in \MB{q^{4t-1}}$, suppose $\bm{x}_s=\SB{x_1,x_2,\cdots,x_{4t-1}}$, let $\V{c}_s=\SB{c_1,c_2,\cdots,c_{2k}}$, $\beta^{(q,4t-1,k)}(\bm{x}_s)=\SB{\beta_1,\beta_2,\cdots,\beta_k}$ for all $1\leq i\leq k$, where $\V{c}_s$ is defined as follows:

\begin{equation}
\begin{cases}
&c_{2i}= (i-1)\Floor{\frac{N}{k}}+1+\SB{\beta_i \mod\Floor{\frac{N}{k}}},\\
&c_{2i-1}=(i-1)\Floor{\frac{N}{k}}+ 1+\Bigl\lfloor{\frac{\beta_i}{\Floor{\frac{N}{k}}}\Bigr\rfloor}.
\end{cases}
\label{eqn: construct1}
\end{equation}

Then $\Se{A}(N,2k,t)=\LB{\V{c}_s:\ s\in\MB{q^{4t-1}}}$ is a $t$-auxiliary set with cardinality $q^{4t-1}$.
\end{theo}

\begin{proof} Without loss of generality, we prove the statement for $\bm{x}_1,\bm{x}_2\in\Db{F}_q^{4t-1}$, $\bm{x}_1\neq \bm{x}_2$, let $\bm{\beta}_1=\beta^{(q,4t-1,k)}(\bm{x}_1)$, $\bm{\beta}_2=\beta^{(q,4t-1,k)}(\bm{x}_2)$. Then, according to \Cref{theo6},
\begin{equation*}
\begin{split}
d_H(\bm{\beta}_1,\bm{\beta}_2)&>\frac{k}{2}-(4t-1)(2+\log_N 2)\\
&>\frac{k}{2}-(12t-3)>\frac{28t}{2}-12t=2t.
\end{split}
\end{equation*} 

In equation (\ref{eqn: construct1}), let $m_i=(i-1)\Floor{\frac{N}{k}}+1$. Notice that $(c_{2i-1}-m_i)\Floor{\frac{N}{k}}+(c_{2i}-m_i)=\beta_i$, for $1\leq i\leq k$. Given $\beta_i\leq N+k$ for all $1\leq i\leq k$, and $k<\Floor{\sqrt{N}-\frac{1}{2}}$, it follows that

\begin{equation*}
\begin{split}
\Bigl\lfloor{\frac{N}{k}\Bigr\rfloor}^2&>\SB{\frac{N}{k}-1}^2\geq \SB{\frac{N}{\sqrt{N}-\frac{3}{2}}-1}^2\\
&>\SB{\sqrt{N}+\frac{3}{2}-1}^2=\SB{\sqrt{N}+\frac{1}{2}}^2\\
&>N+\sqrt{N}>N+k\geq  \beta_i.
\end{split}
\end{equation*} 
Therefore, $\SB{c_{2i-1}-m_i,c_{2i}-m_i}$ is exactly the $\Floor{\frac{N}{k}}$-ary representation of $\beta_i$, for all $1\leq i\leq k$.

Suppose $\bm{\beta}_1=(\beta_{1,1},\beta_{1,2},\cdots,\beta_{1,k})$ and $\bm{\beta}_2=(\beta_{2,1},\beta_{2,2},\cdots,\beta_{2,k})$. Let $Y=\LB{i:\ \beta_{1,i}\neq \beta_{2,i},\ 1\leq i\leq k}$, then $\Nm{Y}=d_H(\bm{\beta}_1,\bm{\beta}_2)$. Notice that for all $i\in Y$, $\beta_{1,i}\neq \beta_{2,i}$, then either $c_{1,2i-1}-m_i\neq c_{2,2i-1}-m_i$ or $c_{1,2i}-m_i\neq c_{2,2i}-m_i$, which means that 
\begin{equation}
\Nm{H(\V{c}_1,\V{c}_2)\cap\LB{c_{1,2i-1},c_{1,2i}}}\geq 1,\ i\in Y.
\label{coro52}
\end{equation}

Notice that $(i-1)\Floor{\frac{N}{k}}<c_{1,2i-1},c_{1,2i}\leq i\Floor{\frac{N}{k}}$, and therefore,
\begin{equation}
\LB{c_{1,2i-1},c_{1,2i}}\cap\LB{c_{1,2i'-1},c_{1,2i'}}=\emptyset,\ \forall\ 1\leq i<i'\leq k.
\label{coro53}
\end{equation}

From (\ref{coro52}) and (\ref{coro53}), 
\begin{equation*}
\begin{split}
\Nm{H(\V{c}_1,\V{c}_2)}&=\Sma{i=1}{k}\Nm{H(\V{c}_1,\V{c}_2)\cap\LB{c_{1,2i-1},c_{1,2i}}}\\
&\geq\Sma{i\in Y}{}\Nm{H(\V{c}_1,\V{c}_2)\cap\LB{c_{1,2i-1},c_{1,2i}}}\\
&\geq \Sma{i\in Y}{}1=\Nm{Y}=d_H(\bm{\beta}_1,\bm{\beta}_2)>2t.
\end{split}
\label{coro51}
\end{equation*}
From \Cref{defi9}, $\Se{A}(N,k,t)$ is indeed a $t$-auxiliary set.
\end{proof}

\begin{rem}\label{rem8} Suppose we use $k=28t$ in \Cref{theo7} to construct a $t$-auxiliary set $\Se{A}(N,56t,t)$. Then the code $\Se{C}_B^{\Tx{sys}}(N,56t,t)$ constructed using \Cref{theo4} based on this $\Se{A}(N,56t,t)$ is an order-optimal systematic $t$-block permutation code.
\end{rem}

\section{Comparison of Cardinality of the Codebooks}
\label{sec: comparison}
In \Cref{sec: non-syscodes}, we constructed a $t$-generalized Cayley code $\Se{C}_{G}(N,t)=\Se{C}_{\bm{\alpha}}(N,4t)$. Let the cardinality of $\Se{C}_{G}(N,t)$ be $A_G(N,t)$. In \cite{6875376}, a $t$-generalized Cayley code with cardinality $A_{\rho_gC}(N,t)$ was constructed. We next compare in \Cref{lemma15} the logarithms of the cardinalities of these two codes, which reflects the redundancy in terms of bits. We show that the proposed scheme requires a smaller number of redundant bits than its counterpart presented in \cite{6875376} for sufficiently large $N$ and $t=o(\frac{N}{\log N})$.

\begin{lemma}\label{lemma15} 
$\log \Nm{A_G(N,t)}>\log \Nm{A_{\rho_gC}(N,t)}$ when $t<\frac{N}{(16\log N+8)}$ for sufficiently large $N$.
\end{lemma}

\begin{proof} 

We know from \cite[Appendix A]{7302598} that:
\begin{equation}
\log \lvert A_{\rho_gC}(N,t)\rvert \leq N\log N-(2+\log e)N+O\left((\log N)^2\right).
\label{eqn:appendix}
\end{equation}
Also,
\begin{equation}
\begin{split}
&\log \lvert A_{G}(N,t)\rvert\\
>& \log N!-\SB{16t(2\log N+1)}\\
>& \SB{N+\frac{1}{2}}\log N-(\log e)N-16t(2\log N+1).\\
\end{split}
\end{equation}
Then,
\begin{equation}
\begin{split}
&\log \Nm{A_{G}(N,t)}-\log \Nm{A_{\rho_gC}(N,t)}\\
>&\SB{N+\frac{1}{2}}\log N-(\log e)N-16t(2\log N+1)\\
&-\SB{N\log N-(2+\log e)N+O\left((\log N)^2\right)}\\
=&\frac{1}{2}\log N+2N-16t(2\log N+1)+O\left((\log N)^2\right)\\
\end{split}
\end{equation}
for sufficiently large $N$ and $t<\frac{N}{(16\log N+8)}$.

From the above discussion, our proposed code in \Cref{sec: non-syscodes} indeed has a higher rate than the interleaving-based code for sufficiently large $N$ and $t=o\SB{\frac{N}{\log N}}$.
\end{proof}

Based on \Cref{rem8} in \Cref{sec: syscodes}, we presented a construction of systematic $t$-generalized Cayley code $\Se{C}'_{G}(N,t)=\Se{C}^{\Tx{sys}}_B(N,56\cdot4t,4t)=\Se{C}^{\Tx{sys}}_B(N,224t,4t)$ with cardinality $A'_G(N,t)$.

In the next \Cref{lemma16}, we compare the logarithm of $A'_{G}(N,t)$ with that of $A_{\rho_gC}(N,t)$.

\begin{lemma}\label{lemma16} 
$A'_{G}(N,t)>A_{\rho_gC}(N,t)$ when $t<\min\LB{\frac{N}{112\log N},\frac{1}{112}\Floor{\sqrt{N}-\frac{1}{2}}}$ for sufficiently large $N$.
\end{lemma}

\begin{proof} 
We know from \Cref{lemma6} that:
\begin{equation}
\begin{split}
\log \Nm{A'_{G}(N,t)}&> (N+\frac{1}{2})\log N-(\log e)N-224t\log N.\\
\label{eqn:cardsys}
\end{split}
\end{equation}
Then it follows from (\ref{eqn:cardsys}) and (\ref{eqn:appendix}) that
\begin{equation}
\begin{split}
&\log \Nm{A'_{G}(N,t)}-\log \Nm{A_{\rho_gC}(N,t)}\\
>&\SB{N+\frac{1}{2}}\log N-(\log e)N-224t\log N\\
&-\SB{N\log N-(2+\log e)N+O\left((\log N)^2\right)}\\
=&\frac{1}{2}\log N+2N-224t\log N+O\left((\log N)^2\right).\\
\end{split}
\end{equation}
for sufficiently large $N$ and $t<\min\LB{\frac{N}{112\log N},\frac{1}{112}\Floor{\sqrt{N}-\frac{1}{2}}}$.

From the above discussion, our proposed systematic code indeed has a higher rate than the interleaving-based code, for sufficiently large $N$ and $t=o\SB{\frac{N}{\log N}}$, in the generalized Cayley distance.
\end{proof}

\section{Conclusion}
\label{sec: conclusion}
The generalized Cayley metric is a distance measure that generalizes the Kendall-tau metric and the Ulam metric. Interleaving was previously shown to be convenient in constructions of permutation codes in the generalized Cayley metric. However, interleaving incurs a noticeable rate penalty such that the constructed permutation codes cannot be order-optimal. In this paper, we presented a framework for constructing order-optimal permutation codes that does not require interleaving. Based on this framework, we then presented an explicit construction of systematic permutation codes from so-called extensions of permutations. We further provided a systematic construction that is order-optimal. Lastly, we proved that our proposed codes are more rate efficient than the existing coding schemes based on interleaving for sufficiently large $N$ and $t=o\SB{\frac{N}{\log N}}$.


\appendices
\section{Proof of \Cref{lemma1}}
\addtocounter{lemma}{-14}
\label{prooflemma1}
\begin{lemma}
For all $\pi_1,\pi_2\in\Db{S}_N$, 
\begin{equation*}
d_B(\pi_1,\pi_2)=\lvert A(\pi_2)\setminus A(\pi_1)\rvert=\lvert A(\pi_1)\setminus A(\pi_2)\rvert.
\end{equation*}
\end{lemma}
\begin{proof} 
According to the symmetry property of the block permutation distance, it is sufficient to prove $d_B(\pi_1,\pi_2)=\lvert A(\pi_1)\setminus A(\pi_2)\rvert$.

Suppose $\pi_1,\pi_2\in\Db{S}_N$ such that $d_B(\pi_1,\pi_2)=d$. Then, there exists $\sigma\in\Db{S}_{d+1}$, $\psi_1,\psi_2,\cdots,\psi_{d+1}$, such that $\pi_1=\left(\psi_1,\psi_2,\cdots,\psi_{d+1}\right)$ and $\pi_2=\left(\psi_{\sigma(1)},\psi_{\sigma(2)},\cdots,\psi_{\sigma(d+1)}\right)$. Suppose $\psi_k=\pi_1\left[i_{k-1}+1:i_k\right]$ for $1\leq k\leq d+1$, where $0=i_0<i_1\cdots<i_d<i_{d+1}=N$. Then $(\pi_1(i),\pi_1(i+1))\in (A(\pi_1)\setminus A(\pi_2))$ if and only if $i\in\{i_1,\cdots,i_d\}$. Therefore, $\lvert A(\pi_1)\setminus A(\pi_2)\rvert=\lvert \{i_1,\cdots,i_d\}\rvert=d$.
\end{proof}

\section{Proof of \Cref{lemma2}}
\label{prooflemma2}
\begin{lemma} For all $\pi_1,\pi_2\in\Db{S}_N$, the following inequality holds,
  \begin{equation*}
    w_B\left(\pi_1\comp\pi_2\right)\leq w_B\left(\pi_1\right)+w_B\left(\pi_2\right).
  \end{equation*}

\end{lemma}

\begin{proof} 

For $\pi\in \Db{S}_N$, define $B(\pi)$ as follows,
\begin{equation*}
 B(\pi)\triangleq \{i|\pi(i+1)\neq\pi(i)+1,\ 1\leq i<N\}.
\end{equation*}
Then, for all $i\in B(\pi)$, $(\pi(i),\pi(i+1))\notin A(e)$. Therefore,
\begin{equation*}
 B(\pi)=\{i|(\pi(i),\pi(i+1))\in \left(A(\pi)\setminus A(e)\right), 1\leq i<N\},
\end{equation*}
which indicates that 
\begin{equation}
\lvert B(\pi)\rvert=\lvert A(\pi)\setminus  A(e)\rvert=w_B(\pi).
\label{Bweight}
\end{equation}

Let $B_1=B(\pi_1)$, $B_2=B(\pi_2)$, $B_3=B(\pi_1\comp\pi_2)$. Then $\forall$ $i\in B_3$, 
\begin{equation*}
\pi_1\left(\pi_2(i+1)\right)\neq \pi_1\left(\pi_2(i)\right)+1.
\end{equation*}
Therefore, $i$ must satisfy at least one of the conditions below:
\begin{equation}
  \begin{split}
    &\{\pi_2(i+1)\neq\pi_2(i)+1\},\ or \\
    &\{\pi_2(i)=k\ and\ \pi_1(k+1)\neq \pi_1(k)+1\}.
  \end{split}
  \label{eqn:tri}
\end{equation}

Equation (\ref{eqn:tri}) means that either $i\in B_2$ or $\pi_2(i)\in B_1$ is true for all $i\in B_3$. Then the function $f: \left(B_3\setminus B_2\right)\to B_1$ specified by $f(i)\triangleq \pi_2(i)$ is an injection, which implies that
\begin{equation}
\lvert B_3\rvert=\lvert B_3 \setminus B_2 \rvert + \lvert B_3 \cap B_2 \rvert \leq \lvert B_1\rvert+\lvert B_2\rvert.
\label{eqn:settri}
\end{equation}
Apply (\ref{Bweight}) to (\ref{eqn:settri}), we obtain the following inequality:
  \[
    w_B\left(\pi_1\comp\pi_2\right)\leq w_B\left(\pi_1\right)+w_B\left(\pi_2\right).
  \]
\end{proof}

\section{Proof of \Cref{lemma4}}
\label{prooflemma4}
\addtocounter{lemma}{1}
\begin{lemma}
For all $N\in\Db{N}^{*}$, $t\leq N-\sqrt{N}-1$, $b_B(N,t)$ is bounded by the following inequality:
\begin{equation*}
\prod\limits_{k=1}^{t}(N-k)\leq b_B(N,t)\leq \prod\limits_{k=0}^{t}(N-k).\\
\end{equation*}

\end{lemma}
\begin{proof} Denote the number of permutations of length $N$ with block permutation weight $m$ by $F(m)$, then $b_B(N,t)=\sum\limits_{m=0}^{t}F(m)$.

We know that $F(0)=1$, and from \cite[equation (3)]{MYERS2002345}, for all $1\leq m\leq t$, 
\begin{equation}
F(m)=\binom{N-1}{m}m!\sum\limits_{k=0}^m(-1)^{m-k}\frac{(k+1)}{(m-k)!}.
\label{eqn:fm}
\end{equation}

Let $a_{k}=\frac{(k+1)}{(m-k)!}$, $0\leq k\leq m$, $1\leq m\leq t$. Then, $m+1=a_m>a_{m-1}=m>a_{m-2}>\cdots>a_0>0$. Therefore, the following inequalities hold true,
\begin{equation*}
\begin{split}
&a_{2k}-a_{2k-1}+\cdots+a_0=a_0+\sum\limits_{i=1}^{k}(a_{2i}-a_{2i-1})>0,\\
&a_{2k-1}-a_{2k-2}+\cdots-a_0=\sum\limits_{i=1}^{k}(a_{2i-1}-a_{2i-2})>0.
\end{split}
\end{equation*}
For $1\leq m\leq t$, define $A_m$ as follows,
\begin{align*}
 A_m=\sum\limits_{k=0}^m(-1)^{m-k}\frac{(k+1)}{(m-k)!}.
\end{align*}
Then, $A_1=1$ and for $2\leq m\leq t$,
\begin{equation}\label{eqn: AAAAAAA}
\begin{split}
A_m&=m+1-(a_{m-1}-a_{m-2}+\cdots+(-1)^{m-1} a_0)<m+1,\\
A_m&=m+1-m+(a_{m-2}-a_{m-3}+\cdots+(-1)^m a_0)>1.
\end{split}
\end{equation}
According to (\ref{eqn:fm}) and (\ref{eqn: AAAAAAA}), for all $1\leq m\leq t$,
\begin{equation*}
\binom{N-1}{m}m!\leq F(m)<\binom{N-1}{m}(m+1)!.
\end{equation*}

To derive the upper bound of the ballsize $b_B(N,t)$, we first find an upper bound of $F(m)$, $1\leq m\leq t$, as follows,
\begin{equation*}
F(m)\leq\binom{N-1}{m}(m+1)!=(m+1)\cdot \prod\limits_{k=1}^{m}(N-k).
\end{equation*}
For $t\leq N-\sqrt{N}-1$, it follows that $i\leq N-\sqrt{N}-1$ for all $1\leq i\leq t$. Therefore, for all $1\leq i\leq t$, 
\begin{equation*}
(N-i-1)^2\geq(N-(N-\sqrt{N}))^2=N>i+1.
\end{equation*}
Then,
\begin{equation*}
\begin{split}
&b_B(N,t)=\sum\limits_{i=0}^{t}F(i)\\
\leq &1+\sum\limits_{i=1}^{t}(i+1)\cdot \prod\limits_{k=1}^{i}(N-k)\\
=&1+\sum\limits_{i=1}^{t}\left(N-(N-i-1)\right)\cdot \prod\limits_{k=1}^{i}(N-k)\\
=&1+\sum\limits_{i=1}^{t}\left(\prod\limits_{k=0}^{i}(N-k)-\prod\limits_{k=1}^{i+1}(N-k)\right)\\
=&\prod\limits_{k=0}^{t}(N-k)-\sum\limits_{i=2}^{t}\left(\prod\limits_{k=1}^{i+1}(N-k)-\prod\limits_{k=0}^{i-1}(N-k)\right)\\
&-(N-1)(N-2)+1\\
=&\prod\limits_{k=0}^{t}(N-k)-\sum\limits_{i=2}^{t}\left(\prod\limits_{k=1}^{i-1}(N-k)\right)\\
&\left((N-i)(N-i-1)-N\right)-((N-1)(N-2)-1)\\
=&\prod\limits_{k=0}^{t}(N-k)-\sum\limits_{i=2}^{t}\left(\prod\limits_{k=1}^{i-1}(N-k)\right)\\
&\left((N-i-1)^2-i-1\right)-((N-1)(N-2)-1)\\
\leq &\prod\limits_{k=0}^{t}(N-k).\\
\end{split}
\end{equation*}

Similarly, for the lower bound, the following inequality holds true.
\begin{equation*}
b_B(N,t)=\sum\limits_{i=0}^{t}F(i)\geq 1+\sum\limits_{i=1}^{t}\prod\limits_{k=1}^{i}(N-k)>\prod\limits_{k=1}^{t}(N-k).
\end{equation*}

The lemma is proved.
\end{proof}

\section{Proof of \Cref{lemma5}}
\label{prooflemma5}
\begin{lemma}
For all $N\in\Db{N}^{*}$, $t\leq \min\{N-\sqrt{N}-1,\frac{N-1}{4}\}$, $b_{G}(N,t)$ is bounded as follows:
\begin{equation*}
\prod\limits_{k=1}^{t}(N-k)\leq b_{G}(N,t)\leq \prod\limits_{k=0}^{4t}(N-k).\\
\end{equation*}
\end{lemma}
\begin{proof}
The upper bound is obtained from replacing $t$ by $4t$ in (\ref{eqn:lemmablk}) and utilizing (\ref{eqn:metricemb}). Note that $\pi\in B_G(N,t,e)$ implies that $d_G(\pi,e)\leq t$. Then from (\ref{eqn:metricemb}), $d_B(\pi,e)\leq 4d_G(\pi,e)\leq 4t$ holds true, which means that $\pi\in B_B(N,4t,e)$. Therefore, $B_G(N,t,e)\subseteq B_B(N,4t,e)$, which implies that $b_G(N,t)\leq b_B(N,4t)$. From (\ref{eqn:lemmablk}) we will get the upper bound. 

Similarly, (\ref{eqn:metricemb}) also implies that $B_B(N,t,e)\subseteq B_G(N,t,e)$, which means that $b_B(N,t)\leq b_G(N,t)$. From (\ref{eqn:lemmablk}) the lower bound follows immediately. The lemma is proved.
\end{proof}

\section{Proof of \Cref{lemma7}}
\label{prooflemma7}
\addtocounter{lemma}{1}
\begin{lemma}
For all $\pi_1$, $\pi_2\in\Db{S}_N$ such that $\pi_1\neq \pi_2$, if $\alpha^{(q,d)}(\pi_1)=\alpha^{(q,d)}(\pi_2)$, then,
\begin{equation*}
\lvert \nu(\pi_1)\Delta \nu(\pi_2) \rvert>2d.
\end{equation*}

\end{lemma}
\begin{proof}
Let $B_1=\nu(\pi_1)$, $B_2=\nu(\pi_2)$. We prove the statement by contradiction. If the lemma is not true, i.e., $\lvert B_1\Delta B_2 \rvert\leq 2d$, then $k=\lvert D_1\rvert=\lvert D_2\rvert\leq d$, where $D_1=B_1\setminus B_2$, $D_2=B_2\setminus B_1$. Suppose $D_1=\{x_1,x_2,\cdots,x_k\}$, $D_2=\{x_{k+1},x_{k+2},\cdots,x_{2k}\}$. Then, $\alpha^{(q,d)}(\pi_1)=\alpha^{(q,d)}(\pi_2)$ is equivalent to the following equations. 

\begin{eqnarray}
\left\{
\begin{array}{ccl}
  x_1+\cdots+x_k&=&x_{k+1}+\cdots+x_{2k},\\
  x_1^2+\cdots+x_k^2&=&x_{k+1}^2+\cdots+x_{2k}^2,\\
  &\vdots&\\
  x_1^{2d-1}+\cdots+x_k^{2d-1}&=&x_{k+1}^{2d-1}+\cdots+x_{2k}^{2d-1}.
\end{array}
\right.
\label{eqn:set}
\end{eqnarray}

From (\ref{eqn:set}), it follows that

\begin{equation*}
\begin{pmatrix}
1 & 1 & \cdots & 1\\
x_1 & x_2 &\cdots &x_{2k}\\
x_1^{2} & x_2^{2} &\cdots &x_{2k}^{2}\\
\vdots &\vdots &\ddots &\vdots\\
x_1^{2d-1} & x_2^{2d-1} &\cdots &x_{2k}^{2d-1}
\end{pmatrix}
\V{y}=\V{0},
\end{equation*}
where $\V{y}=\left(y_1,y_2,\cdots,y_{2k}\right)^T$, and

\begin{equation*}
y_i=
\begin{cases}
1,&1\leq i\leq k,\\
-1,&k<i\leq 2k.
\end{cases}
\end{equation*}

Given that $2k\leq 2d$, the above equation implies that  

\begin{equation}
\begin{pmatrix}
1 & 1 & \cdots & 1\\
x_1 & x_2 &\cdots &x_{2k}\\
x_1^{2} & x_2^{2} &\cdots &x_{2k}^{2}\\
\vdots &\vdots &\ddots &\vdots\\
x_1^{2k-1} & x_2^{2k-1} &\cdots &x_{2k}^{2k-1}
\end{pmatrix}
\V{y}=\V{0}.
\label{matrix}
\end{equation}

Denote the Vandermonde matrix in equation (\ref{matrix}) by $\V{U}$. Then $\V{y}$ is in the nullspace of $\V{U}$. Therefore, $\V{U}$ is singular, which implies that the determinant of $\V{U}$ is equal to $0$ in $\Db{F}_q$, i.e., 

\begin{equation}
0=\det{\V{U}}=\prod\limits_{1\leq i<j\leq 2k}\left(x_i-x_j\right).
\label{det}
\end{equation}

As $q$ is a divisor of $0$, $q$ should also be a divisor of the right hand side of equation (\ref{det}), which implies that $\exists$ $i\neq j\in\left[2k\right]$ such that $q|(x_i-x_j)$. Then $x_i=x_j$ on $\Db{F}_q$, and we must have $x_i\in D_1,x_j\in D_2$ or $x_i\in D_2,x_j\in D_1$, which implies that $x_i,x_j\in D_1\cap D_2$, a contradiction.
\end{proof}

\section{Proof of \Cref{lemma8}}
\label{prooflemma8}
\begin{lemma} Suppose $\V{A}\in \Db{F}_q^{(4t-1)\times(2t)}$, $\V{b}\in \Db{F}_q^{4t-1}$ are defined in (\ref{eqn:a}) and (\ref{eqn:b}), respectively. Consider the following equation defined on $\Db{F}_q$:
\begin{equation*}
\V{A}\V{c}=\V{b}.
\end{equation*}

For any vector $\V{c}\in\Db{F}_q^{2t}$, $\V{c}$ is a nonzero solution to (\ref{eqn:dec}) if and only if $(h_1(\V{c}),h_2(\V{c}))$ is a nonzero solution to (\ref{eqn:decode}).
\end{lemma}
\begin{proof}
Suppose 
\begin{equation}
\begin{split}
f_1&=X^{N-1}+a_1 X^{N-2}+\cdots+a_{4t-1} X^{N-4t}+g_1,\\
f_2&=X^{N-1}+a'_1 X^{N-2}+\cdots+a'_{4t-1} X^{N-4t}+g_2.\\
\end{split}
\label{eqn:f}
\end{equation}
Additionally, suppose
\begin{equation*}
\begin{split}
h_1\cdot f_1&=X^{N+t-1}+s_{N+t-2}X^{N+t-2}+\cdots+s_0,\\
h_2\cdot f_2&=X^{N+t-1}+s'_{N+t-2}X^{N+t-2}+\cdots+s'_{0}.
\end{split}
\end{equation*}
Then, from (\ref{eqn:f}) and (\ref{eqn:h}), it follows that

\begin{equation*}
\begin{cases}
& s_{N+t-2}=a_1+c_1,\\
& s_{N+t-3}=a_2+c_1a_1+c_2,\\
&\vdots\\
& s_{N-1}=a_{t}+c_1a_{t-1}+\cdots+c_{t},\\
&\vdots\\
& s_{N-3t}=a_{4t-1}+c_1a_{4t-2}+c_2a_{4t-3}+\cdots+c_{t}a_{3t-1}.
\end{cases}
\end{equation*}
Similarly, we also have

\begin{equation*}
\begin{cases}
& s'_{N+t-2}=a'_1+c'_1,\\
& s'_{N+t-3}=a'_2+c'_1a'_1+c'_2,\\
&\vdots\\
& s'_{N-1}=a'_{t}+c'_1a'_{t-1}+\cdots+c'_{t},\\
&\vdots\\
& s'_{N-3t}=a'_{4t-1}+c'_1a'_{4t-2}+c'_2a'_{4t-3}+\cdots+c'_{t}a'_{3t-1}.
\end{cases}
\end{equation*}
Then (\ref{eqn:decode}) is true iff $ s_i=s'_i$ for all $N-3t \leq i\leq N+t-2$, which is equivalent to the following equation:

\begin{equation}
\begin{split}
&
\begin{pmatrix}
1        & &&\\
a_1      & 1        & &\\
\vdots     &\vdots     & \ddots&\\
a_{t-1}      & a_{t-2}      &  \cdots &   1\\
\vdots   & \vdots   &  \ddots &  \vdots\\
a_{4t-2} & a_{4t-3} &  \cdots& a_{3t-1}
\end{pmatrix}
\begin{pmatrix}
c_1\\
c_2\\
\vdots\\
c_{t}
\end{pmatrix}
+
\begin{pmatrix}
a_1\\
a_2\\
\vdots\\
a_{4t-1}
\end{pmatrix}
=\\
&
\begin{pmatrix}
1        & &&\\
a'_1      & 1        & &\\
\vdots      & \vdots    & \ddots&\\
a'_{t-1}      & a'_{t-2}      & \cdots  &   1\\
\vdots   & \vdots   &  \ddots &  \vdots\\
a'_{4t-2} & a'_{4t-3} & \cdots & a'_{3t-1}
\end{pmatrix}
\begin{pmatrix}
c'_1\\
c'_2\\
\vdots\\
c'_{t}
\end{pmatrix}
+
\begin{pmatrix}
a'_1\\
a'_2\\
\vdots\\
a'_{4t-1}
\end{pmatrix}
.
\end{split}
\label{eqn:matrix}
\end{equation}
We note that (\ref{eqn:matrix}) is equivalent to (\ref{eqn:dec}). 
\end{proof}

\section{Proof of \Cref{lemma9}}
\label{prooflemma9}
\begin{lemma} Let $\pi_1,\pi_2\in \Db{S}_N$, $s_1,s_2\in\MB{N}$. For any two extensions $E(\pi_1,s_1)$ and $E(\pi_2,s_2)$, if $s_1$ is a jump point of $E(\pi_1,s_1)$ with respect to $E(\pi_2,s_2)$, then
\begin{equation*}
d_B(E(\pi_1,s_1),E(\pi_2,s_2))>d_B(\pi_1,\pi_2),
\end{equation*}
else
\begin{equation*}
d_B(E(\pi_1,s_1),E(\pi_2,s_2))=d_B(\pi_1,\pi_2).
\end{equation*}

\end{lemma} 
\begin{proof} Let $\sigma_1=E(\pi_1,s_1)$ and $\sigma_2=E(\pi_2,s_2)$. Recall the notion of \emph{characteristic sets} in \Cref{defi3}. Suppose $A(\pi_1)$, $A(\pi_2)$, $A(\sigma_1)$, $A(\sigma_2)$ are the characteristic sets of $\pi_1$, $\pi_2$, $\sigma_1$, $\sigma_2$, respectively. According to \Cref{lemma1},
\begin{equation}
\begin{split}
d_B(\pi_1,\pi_2)&=\lvert A(\pi_1)\setminus A(\pi_2) \rvert,\\
d_B(\sigma_1,\sigma_2)&=\lvert A(\sigma_1)\setminus A(\sigma_2) \rvert.
\end{split}
\label{extensioncode}
\end{equation}

Let $k_1=\pi_1^{-1}(s_1)$, $k_2=\pi_2^{-1}(s_2)$, then $\pi_{1}(k_1)=s_1$ and $\pi_{2}(k_2)=s_2$. If $1\leq k_1,k_2<N$, let $\pi_{1}(k_1+1)=j_1$ and $\pi_{2}(k_2+1)=j_2$. 

Suppose first $s_1$ is a jump point, then consider the following cases.

\begin{enumerate}
\item $s_1\neq s_2$ and either $k_1=N$ or $k_2=N$.
\begin{enumerate}
\item $k_1=k_2=N$. In this case, $A(\sigma_1)=A(\pi_1)\cup \LB{(s_1,N+1)}$, $A(\sigma_2)=A(\pi_2)\cup \LB{(s_2,N+1)}$. Therefore, $A(\sigma_1)\setminus A(\sigma_2)=\SB{A(\pi_1)\setminus A(\pi_2)}\cup\LB{(s_1,N+1)}$. From (\ref{extensioncode}), $d_B(\sigma_1,\sigma_2)=d_B(\pi_1,\pi_2)+1$ follows.

\item $k_1=N\neq k_2$. In this case, $A(\sigma_1)=A(\pi_1)\cup \LB{(s_1,N+1)}$, $A(\sigma_2)=\SB{A(\pi_2)\setminus\LB{(s_2,j_2)}}\cup \LB{(s_2,N+1),(N+1,j_2)}$. Therefore, $A(\sigma_1)\setminus A(\sigma_2)=\SB{A(\pi_1)\setminus \SB{A(\pi_2)\setminus\LB{(s_2,j_2)}}}\cup\LB{(s_1,N+1)}$, which means $\SB{\SB{A(\pi_1)\setminus A(\pi_2)}\cup\LB{(s_1,N+1)}}\subseteq \SB{A(\sigma_1)\setminus A(\sigma_2)}$. From (\ref{extensioncode}), it follows that $d_B(\sigma_1,\sigma_2)\geq d_B(\pi_1,\pi_2)+1$.

\item $k_2=N\neq k_1$. Following the same logic in the previous case, $d_B(\sigma_1,\sigma_2)\geq d_B(\pi_1,\pi_2)+1$ holds true.
\end{enumerate}

\item $s_1\neq s_2$, $k_1,k_2\neq N$. Since $s_1$ is a jump point, $j_1\neq j_2$.

\begin{enumerate}
\item In this case, $A(\sigma_1)=\SB{A(\pi_1)\setminus \LB{(s_1,j_1)}}\cup \LB{(s_1,N+1),(N+1,j_1)}$, $A(\sigma_2)=\SB{A(\pi_2)\setminus \LB{(s_2,j_2)}}\cup \LB{(s_2,N+1),(N+1,j_2)}$. Therefore, the equation $(\SB{\SB{A(\pi_1)\setminus A(p_2)}\setminus\LB{s_1,j_1}}\cup\LB{(s_1,N+1),(N+1,j_1)})\subseteq \SB{A(\sigma_1)\setminus A(\sigma_2)}$ follows. From (\ref{extensioncode}), $d_B(\sigma_1,\sigma_2)\geq d_B(\pi_1,\pi_2)+1$.
\end{enumerate}
\end{enumerate}

If $s_1$ is not a jump point, then consider the following cases.

\begin{enumerate}
\item $s_1=s_2$ and either $k_1=N$ or $k_2=N$.
\begin{enumerate}
\item $k_1=k_2=N$. In this case, $A(\sigma_1)=A(\pi_1)\cup \LB{(s_1,N+1)}$, $A(\sigma_2)=A(\pi_2)\cup \LB{(s_1,N+1)}$. Therefore, $A(\sigma_1)\setminus A(\sigma_2)=A(\pi_1)\setminus A(\pi_2)$. From (\ref{extensioncode}), $d_B(\sigma_1,\sigma_2)=d_B(\pi_1,\pi_2)$ follows.

\item $k_1=N\neq k_2$. In this case, $A(\sigma_1)=A(\pi_1)\cup \LB{(s_1,N+1)}$, $A(\sigma_2)=\SB{A(\pi_2)\setminus\LB{(s_1,j_2)}}\cup \LB{(s_1,N+1),(N+1,j_2)}$. Therefore, $A(\sigma_1)\setminus A(\sigma_2)=A(\pi_1)\setminus \SB{A(\pi_2)\setminus\LB{(s_1,j_2)}}=A(\pi_1)\setminus A(\pi_2)$. From (\ref{extensioncode}), it follows that $d_B(\sigma_1,\sigma_2)= d_B(\pi_1,\pi_2)$.

\item $k_2=N\neq k_1$. Follow the same logic in the previous case, $d_B(\sigma_1,\sigma_2)= d_B(\pi_1,\pi_2)$ holds true.
\end{enumerate}

\item $k_1,k_2\neq N$. Since $s_1$ is not a jump point, either $s_1=s_2$ or $j_1=j_2$ must be satisfied.
\begin{enumerate}
\item $s_1=s_2$ and $j_1=j_2$. In this case, $A(\sigma_1)=\SB{A(\pi_1)\setminus \LB{(s_1,j_1)}}\cup \LB{(s_1,N+1),(N+1,j_1)}$, $A(\sigma_2)=\SB{A(\pi_2)\setminus \LB{(s_1,j_1)}}\cup \LB{(s_1,N+1),(N+1,j_1)}$. Therefore, $A(\sigma_1)\setminus A(\sigma_2)=A(\pi_1)\setminus A(\pi_2)$. From (\ref{extensioncode}), $d_B(\sigma_1,\sigma_2)=d_B(\pi_1,\pi_2)$ follows.

\item $s_1=s_2$ and $j_1\neq j_2$. In this case, $A(\sigma_1)=\SB{A(\pi_1)\setminus \LB{(s_1,j_1)}}\cup \LB{(s_1,N+1),(N+1,j_1)}$, $A(\sigma_2)=\SB{A(\pi_2)\setminus \LB{(s_1,j_2)}}\cup \LB{(s_1,N+1),(N+1,j_2)}$. Therefore, $A(\sigma_1)\setminus A(\sigma_2)=\SB{\SB{A(\pi_1)\setminus A(\pi_2)}\setminus\LB{(s_1,j_1)}}\cup\LB{(N+1,j_1)}$. From (\ref{extensioncode}), it follows that $d_B(\sigma_1,\sigma_2)=d_B(\pi_1,\pi_2)$.

\item $s_1\neq s_2$ and $j_1=j_2$. Follow the same logic as indicated in the previous case, $d_B(\sigma_1,\sigma_2)=d_B(\pi_1,\pi_2)$ holds true. 
\end{enumerate}
\end{enumerate}
The lemma is proved.
\end{proof}

\section{Proof of \Cref{lemma10}}
\label{prooflemma10}
\begin{lemma} Let $\pi_1,\pi_2\in \Db{S}_N$, $s_1,s_2\in\MB{N}$. For any extensions $E(\pi_1,S_1)$, $E(\pi_2,S_2)$ of $\pi_1$, $\pi_2$ on extension sequences $S_1$, $S_2$, respectively, it follows that
\begin{equation*}
d_B(E(\pi_1,S_1),E(\pi_2,S_2))\geq \lvert H(S_1,S_2)\rvert.
\end{equation*}

\end{lemma}
\begin{proof} For all $i\in H(S_1,S_2)$, let 

\begin{equation}
m(i)=\min \LB{m:\ s_{1,m}=i,\ s_{2,m}\neq i}.
\label{eqn: lemma10} 
\end{equation}

Suppose $J_{1,m(i)-1}=(s_{1,1},s_{1,2},\cdots,s_{1,m(i)-1})$, $J_{2,m(i)-1}=(s_{2,1},s_{2,2},\cdots,s_{2,m(i)-1})$. Let $\sigma^{m(i)-1}_1=E(\pi_1,J_{1,m(i)-1})$ and $\sigma^{m(i)-1}_2=E(\pi_2,J_{2,m(i)-1})$. Recall the definition of the \emph{jump set} $F(\pi_1,\pi_2,S_1,S_2)$ in \Cref{defi7}. Consider the following two cases:

\begin{enumerate}
\item If $m(i)\in F(\pi_1,\pi_2,S_1,S_2)$, then $s_{1,m(i)}=i$ is a jump point of $E(\sigma^{m(i)-1}_1,s_{1,m(i)})$ with respect to $E(\sigma^{m(i)-1}_2,s_{2,m(i)})$.

\item If $m(i)\notin F(\pi_1,\pi_2,I_1,I_2)$, then $i$ is not a jump point of $E(\sigma^{m(i)-1}_1,s_{1,m(i)})$ with respect to $E(\sigma^{m(i)-1}_2,s_{2,m(i)})$. Let $k'_1=(\sigma^{m(i)-1}_1)^{-1}(s_{1,m(i)})$, $k_1=\pi_1^{-1}(s_{1,m(i)})$, $k'_2=(\sigma^{m(i)-1}_2)^{-1}(s_{2,m(i)})$, $k_2=\pi_2^{-1}(s_{2,m(i)})$, then $\sigma^{m(i)-1}_{1}(k'_1)=\pi_{1}(k_1)=s_{1,m(i)}$ and $\sigma^{m(i)-1}_{2}(k'_2)=\pi_{2}(k_2)=s_{2,m(i)}$. Given that $s_{1,m(i)}$ is not a jump point and $s_{1,m(i)}\neq s_{2,m(i)}$, it follows from \Cref{defi6} that $k_1,k_2\neq N+m(i)-1$ and $\sigma^{m(i)-1}_{1}(k'_1+1)=\sigma^{m(i)-1}_{2}(k'_2+1)$ must be true. Let $j=\sigma^{m(i)-1}_{1}(k'_1+1)=\sigma^{m(i)-1}_{2}(k'_2+1)$. From (\ref{eqn: lemma10}), $\pi_{1}(k_1+1)=\pi_{2}(k_2+1)=j\in\MB{N}$ holds, otherwise $N<j<N+m(i)$ is inserted after $i$ in $\pi_1$ and is not inserted after $i$ in $\pi_2$, a contradiction. Then $(i,j)\in A(\pi_1)$, $(s_{2,m(i)},j)\in A(\pi_2)$ and $s_{2,m(i)}\neq i$. Therefore $(i,j)\in \SB{A(\pi_1)\setminus A(\pi_2)}$.
\end{enumerate}

Suppose $J=\LB{i|m(i)\notin F(\pi_1,\pi_2,S_1,S_2),i\in H(S_1,S_2)}$, then from the above discussion:
\begin{equation*}
\begin{split}
&\lvert F(\pi_1,\pi_2,S_1,S_2)\rvert\geq \lvert H(S_1,S_2)\setminus J \rvert,\\
&d_B(\pi_1,\pi_2)=\lvert A(\pi_1)\setminus A(\pi_2) \rvert\geq \lvert J\rvert.
\end{split}
\end{equation*} 
And from \Cref{lemma9}, it follows that

\begin{equation*}
\begin{split}
d_B(E(\pi_1,S_1),E(\pi_2,S_2))\geq &d_B(\pi_1,\pi_2)+\lvert F(\pi_1,\pi_2,S_1,S_2)\rvert\\
\geq &\lvert H(S_1,S_2)\setminus J \rvert+\lvert J\rvert\\
\geq &\lvert H(S_1,S_2)\rvert.
\end{split}
\end{equation*}
The lemma is proved.
\end{proof}

\section{Proof of \Cref{lemma14}}
\addtocounter{lemma}{1}
\label{prooflemma14}
\begin{lemma} For all $k,N\in\Db{N}^{*}$, $k>3$, $N>k^2$, consider an arbitrary subset $Y\subset \MB{k}$, where $\Nm{Y}=M<k$, $Y=\LB{i_1,i_2,\cdots,i_M}$, then 
\begin{equation*}
\Tx{LCM}\SB{N+i_1,N+i_2,\cdots,N+i_M}>N^{M-\frac{k}{2}}.
\end{equation*}
\end{lemma}
\begin{proof} For all $r,n\in \Db{N}^{*}$, it follows from \cite[equation (13)]{LCM} that
\begin{equation}
\label{eqngr}
g_r(n)=\Tx{GCD}(r!,(n+r)g_{r-1}(n)),
\end{equation}
where for all $r\in \Db{N}$, $n\in\Db{N}^{*}$,
\begin{equation}
\label{eqngr2}
g_r(n)=\frac{n(n+1)\cdots(n+r)}{\Tx{LCM}(n,n+1,\cdots,n+r)}.
\end{equation}
From (\ref{eqngr}) and (\ref{eqngr2}), the following statement holds true,
\begin{equation}
g_r(n)|r!,\ \forall r,n\in\Db{N}^{*},
\end{equation}
which implies that
\begin{equation}
\label{eqnLCMupperbound}
\frac{n(n+1)\cdots(n+r)}{\Tx{LCM}(n,n+1,\cdots,n+r)}\leq r!.
\end{equation}

Let $n=N+1$, $r=k-1$ in (\ref{eqnLCMupperbound}). Then, for all $N,k\in\Db{N}^{*}$,
\begin{equation}
\begin{split}
&\Tx{LCM}\SB{N+1,N+2,\cdots,N+k}\\
\geq &\frac{(N+1)(N+2)\cdots(N+k)}{(k-1)!}.\\
\end{split}
\label{coro41}
\end{equation} 
Let $\MB{k}\setminus Y=\LB{j_1,j_2,\cdots,j_{k-M}}$. Notice that 
\begin{equation}
\begin{split}
&\Tx{LCM}\SB{N+1,N+2,\cdots,N+k}\\
=&\Tx{LCM}(\Tx{LCM}\SB{N+i_1,N+i_2,\cdots,N+i_M},\\
&\Tx{LCM}\SB{N+j_1,N+j_2,\cdots,N+j_{k-M}})\\
\leq &\MB{\Prd{s=1}{k-M}(N+j_{s})}\Tx{LCM}\SB{N+i_1,N+i_2,\cdots,N+i_M}.
\end{split}
\label{coro42}
\end{equation} 
From equation (\ref{coro41}) and (\ref{coro42}), 
\begin{equation*}
\begin{split}
&\Tx{LCM}\SB{N+i_1,N+i_2,\cdots,N+i_M}\\
\geq&\frac{\Tx{LCM}(N+1,N+2,\cdots,N+k)}{\Prd{s=1}{k-M}(N+j_{s})}\\
\geq &\frac{(N+1)(N+2)\cdots(N+k)}{(k-1)!\Prd{s=1}{k-M}(N+j_{s})}=\frac{\Prd{s=1}{M}(N+i_{s})}{(k-1)!}>\frac{N^{M}}{k!}.\\
\end{split}
\end{equation*} 

From \Cref{lemma6}, for all $k>3$ and $N>k^2$,
\begin{equation*}
\begin{split}
\frac{N^{M}}{k!}&>\frac{N^{M}}{2^{(k+\frac{1}{2})\log k -k+2}}>\frac{N^{M} 2^{k-2}}{k^{k+1}}\geq\frac{N^{M}}{k^k}>N^{M-\frac{k}{2}}.
\end{split}
\end{equation*} 

The lemma is proved.
\end{proof}

\section*{Acknowledgment}

Research supported in part by NSF CAREER grant No.1150212 and NSF CIF grant No.1527130.

\bibliography{ref}

\bibliographystyle{IEEEtran}

\section*{Biographies}

\textbf{Siyi Yang} (S'17) is a Ph.D. candidate in the Electrical and Computer Engineering department at the University of California, Los Angeles (UCLA). She received her B.S. degree in Electrical Engineering from the Tsinghua University, in 2016 and the M.S. degree in Electrical and Computer Engineering from the University of California, Los Angeles (UCLA) in 2018. Her research interests include design of error-correction codes for non-volatile memory and distributed storage.

\textbf{Clayton Schoeny} (S'09) is a data scientist working at Fair. He received his Ph.D. in the Electrical and Computer Engineering Department at the University of California, Los Angeles (UCLA) in 2018. He received his B.S. and M.S. degrees in Electrical Engineering from UCLA in 2012 and 2014, respectively. He is a recipient of the Henry Samueli Excellence in Teaching Award, the 2016 Qualcomm Innovation Fellowship, and the 2017 UCLA Dissertation Year Fellowship.

\textbf{Lara Dolecek} (S'05--M'10--SM'13) is a Full Professor with the Electrical and Computer Engineering Department and Mathematics Department (courtesy) at the University of California, Los Angeles (UCLA). She holds a B.S. (with honors), M.S., and Ph.D. degrees in Electrical Engineering and Computer Sciences, as well as an M.A. degree in Statistics, all from the University of California, Berkeley. She received the 2007 David J. Sakrison Memorial Prize for the most outstanding doctoral research in the Department of Electrical Engineering and Computer Sciences at UC Berkeley. Prior to joining UCLA, she was a postdoctoral researcher with the Laboratory for Information and Decision Systems at the Massachusetts Institute of Technology. She received IBM Faculty Award (2014), Northrop Grumman Excellence in Teaching Award (2013), Intel Early Career Faculty Award (2013), University of California Faculty Development Award (2013), Okawa Research Grant (2013), NSF CAREER Award (2012), and Hellman Fellowship Award (2011). With her research group and collaborators, she received numerous best paper awards. Her research interests span coding and information theory, graphical models, statistical methods, and algorithms, with applications to emerging systems for data storage and computing. She currently serves as an Associate Editor for IEEE Transactions on Communications. Prof. Dolecek has served as a consultant for a number of companies specializing in data communications and storage.

\end{document}